\documentclass[reprint, amsmath, amssymb, nofootinbib, aps, pra]{revtex4-2}

\usepackage{graphicx}
\usepackage{dcolumn}
\usepackage{bm}
\usepackage[colorlinks=true, linkcolor=blue, citecolor=blue]{hyperref}

\usepackage[usenames, dvipsnames]{xcolor}
\usepackage{pifont}
\usepackage{enumitem}
\usepackage{amsthm}
\usepackage{tikz}
\usepackage{tcolorbox}
\usepackage[normalem]{ulem}
\usepackage{booktabs}
\usepackage{mathrsfs}
\usepackage[caption=false]{subfig}

\newcommand{\yes}{\textcolor{dkgreen}{\checkmark}}
\newcommand{\no}{\textcolor{red}{\ding{55}}}
\newcommand\norm[1]{\left\lVert#1\right\rVert}
\newcommand\ip[1]{\left\langle#1\right\rangle}
\newcommand\ket[1]{\left|#1\right\rangle}

\newcommand\op[2]{\left|#1\right\rangle\left\langle#2\right|}
\newcommand\id{\mathbb{I}}

\DeclareMathOperator\tr{Tr}
\renewcommand{\Re}{\textnormal{Re}}
\DeclareMathOperator*{\Prob}{Prob}
\DeclareMathOperator{\affh}{AffH}
\definecolor{dkgreen}{rgb}{0,0.6,0}

\theoremstyle{definition}
\newtheorem{definition}{Definition}[section]
\theoremstyle{plain}

\newtheorem{theorem}{Theorem}[section]
\newtheorem{corollary}{Corollary}[theorem]
\newtheorem*{theorem*}{Theorem}
\newtheorem{proposition}{Proposition}[section]

\newcommand*{\defeq}{\mathrel{\vcenter{\baselineskip0.5ex \lineskiplimit0pt
                     \hbox{\scriptsize.}\hbox{\scriptsize.}}}%
                     =}

\newcommand{\error}{\epsilon}
\newcommand{\failure}{\delta}
\newcommand{\epsZero}{\epsilon_{\circ}}
\newcommand{\jnrisk}{\varepsilon_*}
\newcommand{\adderr}{\varepsilon}
\newcommand{\qsverr}{\varepsilon_{\text{q}}}
\newcommand{\epsSDP}{\epsilon_m}

\begin{document}

\title{Theory of versatile fidelity estimation with confidence}

\author{Akshay Seshadri}
\affiliation{Department of Physics, University of Colorado Boulder, Boulder, USA}
\author{Martin Ringbauer}
\affiliation{Universit\"{a}t Innsbruck, Institut f\"{u}r Experimentalphysik, Technikerstrasse 25, 6020 Innsbruck, Austria}
\author{Jacob Spainhour}
\affiliation{Department of Applied Mathematics, University of Colorado Boulder, Boulder, USA }
\author{Thomas Monz}
\affiliation{Universit\"{a}t Innsbruck, Institut f\"{u}r Experimentalphysik, Technikerstrasse 25, 6020 Innsbruck, Austria}
\affiliation{Alpine Quantum Technologies GmbH, 6020 Innsbruck, Austria}
\author{Stephen Becker}
\affiliation{Department of Applied Mathematics, University of Colorado Boulder, Boulder, USA}

\date{\today}

\begin{abstract}
    Estimating the fidelity with a target state is important in quantum information tasks. Many fidelity estimation techniques present a suitable measurement scheme to perform the estimation. In contrast, we present techniques that allow the experimentalist to choose a convenient measurement setting. Our primary focus lies on a method that constructs an estimator with nearly minimax optimal confidence intervals for any specified measurement setting. We demonstrate, through a combination of theoretical and numerical results, various desirable properties of the method: robustness against experimental imperfections, competitive sample complexity, and accurate estimates in practice. We compare this method with Maximum Likelihood Estimation and the associated Profile Likelihood method, a Semi-Definite Programming based approach, direct fidelity estimation, quantum state verification, and classical shadows. Our method can also be used for estimating the expectation value of any observable with the same guarantees.
\end{abstract}

\maketitle

\section{Introduction}
Quantum information, though a relatively nascent branch of research, has pervaded into many areas of physics and even other subjects like computer science. This prevalence is owed to a large part to the wide range of applications that benefit from a quantum approach, such as computation~\cite{ekert1996quantum}, simulation~\cite{Georgescu2014Review}, communication~\cite{gisin2007quantum}, and metrology and sensing~\cite{degen2017quantum, giovannetti2011advances}. All of these applications require us to accurately and reliably prepare a desired quantum state. Yet, \emph{verifying} that the experimentally prepared state is indeed what we intended is a non-trivial task. A commonly used metric to judge the success of state preparation is the fidelity between the target state $\rho$ and the prepared state $\sigma$. When the target state is pure, the fidelity is simply given as $F(\rho, \sigma) = \tr(\rho \sigma)$~\cite{Nielsen2000}. While we focus on fidelity estimation here, our method can be used to estimate the expectation value of any observable.

A common way to estimate the fidelity is to first reconstruct the experimentally prepared quantum state and then compute the fidelity between the reconstruction and the target state. Quantum state tomography is an active area of research, and several methods have been proposed to reconstruct the quantum state \cite{Vogel1989,Hradil1997, Blume-Kohout2010BME, Gross2010, Cramer2010}. However, these methods suffer from an exponential resource requirement as the Hilbert space grows, making them infeasible for all but the smallest systems. Particularly, when the task is merely to estimate the fidelity, a method known as Direct Fidelity Estimation (DFE)~\cite{Flammia2011, DaSilva2011} has been shown to achieve the task with exponentially fewer resources. DFE and more recent approaches like classical shadows~\cite{Huang2020} estimate the fidelity directly from measurement outcomes without going through the intermediate step of reconstructing the state. A related, but slightly different approach, is quantum state verification~\cite{pallister2018optimal}, where one \textit{tests} whether the fidelity is larger than a given threshold. These methods specify measurement protocols, typically involving sampling measurement settings randomly, that one needs to follow in order to estimate the fidelity or perform state verification. In practice, however, experimenters tend to use a fixed set of measurement settings instead of randomly sampling them, which would result in a loss of theoretical guarantees. Furthermore, there can be cases where it is preferable to implement a different measurement scheme than the ones prescribed by these approaches.
This is particularly true if one wishes to the estimate the expectation value of a single observable, where performing measurements tailored to that observable can give exponential improvements over a method like classical shadows that performs random measurements independent of the observable to be estimated.

On the other hand, methods like maximum likelihood estimation (MLE) work for arbitrary measurement schemes, yet, in addition to the need for a costly reconstruction of the quantum state, they do not provide any guarantees on the estimate. Rigorous confidence intervals for estimated fidelities are crucial for applications from quantum error correction~\cite{Gottesman1997} to quantum key distribution, where the security of a protocol rests on proving the presence of entanglement~ \cite{curty2004entanglement, van2007experimental}. Entanglement verification in turn is typically achieved using an entanglement witness~\cite{brandao2005quantifying, bourennane2004experimental}, which is intimately tied to fidelity estimation~\cite{bourennane2004experimental}, again highlighting the need for rigorous confidence intervals. Even in cases where state reconstruction with a good confidence interval is possible, such as a bound in trace distance as given in Ref.~\cite{guctua2020FastTomography}, the number of measurements needed for tomography can be much larger than what is needed for fidelity estimation~\cite{Flammia2011}.
One can also use quantum algorithms to estimate the fidelity, for example, Ref.~\cite{cerezo2020variational, gilyen2022improved}, which work with a different premise than what we consider in our study.
To summarize, we have identified two concerns with standard approaches for fidelity estimation: most do not allow for arbitrary measurement settings, and/or do not provide rigorous confidence intervals.

Here, we address both of these concerns. We provide a versatile approach for estimating the fidelity directly from raw data for \emph{any} measurement scheme, without the need for state reconstruction. More precisely, given any target state, measurement scheme, and desired confidence level, our method constructs an estimator for the fidelity that takes raw data and gives a fidelity estimate almost instantaneously. Furthermore, the estimator comes equipped with a rigorous confidence interval that is guaranteed to be close to minimax optimal.\footnote{Larger at most by a factor of order 1 for sufficiently large confidence levels; see section~\ref{secn:minimax_method_theory}.} Turning this around, for a target confidence level $1 - \failure$, our method achieves a confidence interval of size $2\jnrisk$ --- where $\jnrisk$ is called the \emph{risk} (or additive error) of the estimate --- with a sample complexity of $\approx \ln(2/\failure)/(2\jnrisk^2)$ independent of the target state $\rho$ when measuring in the basis defined by $\rho$ (see section \ref{secn:minimax_method_optimal_risk}). While such a measurement scheme may often be impractical, we will introduce an alternative scheme using Pauli measurements that achieves a similar sample complexity for stabilizer states. This finding demonstrates that our method can give a scalable sample complexity with a judicious choice of measurement protocol.

We will begin by describing the theory behind our approach in section \ref{secn:minimax_method_theory}. Then, we compute the sample complexity of the method in section \ref{secn:minimax_method_optimal_risk} and describe a Pauli measurement scheme similar to DFE (in procedure and sample complexity) in section \ref{secn:minimax_method_RPM_scheme}. In section \ref{secn:minimax_method_robustness}, we demonstrate the robustness of the estimator generated by our method against noise. Finally, we compare our method with DFE, MLE, profile likelihood, a semidefinite programming approach, quantum state verification, and classical shadows in section \ref{secn:other_methods_comparison}.

\section{Minimax method\label{secn:minimax_method}}
Our approach is based on recent results in statistics by Juditsky \& Nemirovski~\cite{Juditsky2009,goldenshluger2015hypothesis,juditsky2018near}, who describe the estimation of linear functionals in a general setting. The risk associated with the estimate is nearly minimax optimal, and so we call it the minimax method. Roughly speaking, a minimax optimal estimator gives the smallest possible symmetric confidence interval in the worst possible scenario (this will be made precise later).

In section \ref{secn:minimax_method_theory}, we elaborate on the main ideas involved in estimating fidelity using the minimax method. An abstract version of the method with the theoretical details is given in Appendix~\ref{app:minimax_theory}, while an application-oriented presentation of the key results is given in the accompanying Ref.~\cite{PRL}.

\subsection{Fidelity estimation using the minimax method\label{secn:minimax_method_theory}}
Suppose that we are working with a quantum system that has a $d$-dimensional Hilbert space over $\mathbb{C}$, $d \in \mathbb{N}$. The states of this system can be described by density matrices, which are $d \times d$ positive semidefinite (complex-valued) matrices with unit trace. We denote the set of density matrices by $\mathcal{X}$. Our goal is to estimate the fidelity between the state $\sigma \in \mathcal{X}$ prepared in the lab and a pure target state $\rho \in \mathcal{X}$, where pure means that $\rho$ is a rank-one matrix. In this case, the fidelity between $\rho$ and $\sigma$ is given as $F(\rho, \sigma) = \tr(\rho \sigma)$, which is a linear function of $\sigma$ for a fixed $\rho$.

In practice, one needs to estimate the fidelity from partial information about the state obtained through measurements. Any quantum measurement can be described by a positive operator-valued measure (POVM), which comprises a set of positive semidefinite matrices that sum to the identity. We consider the case where an experimenter uses $L$ different measurement settings, where the POVM for the $l^\text{th}$ setting ($l = 1, \dotsc, L$), is described by $\{E^{(l)}_1, \dotsc, E^{(l)}_{N_l}\}$. The experimenter performs $R_l$ repetitions (shots) of the $l^\text{th}$ POVM. The probability $p_\sigma^{(l)}(k)$ that a particular outcome $k \in \{1, \dotsc, N_l\}$ is obtained upon measuring the $l^{\text{th}}$ POVM when the state of the system is $\sigma$ is given by Born's rule as $p_\sigma^{(l)}(k) = \tr(E^{(l)}_k \sigma)$.

Note that Born's rule can give zero probabilities for some of the outcomes depending on the state and the POVM. On the other hand, a technical requirement of Juditsky \& Nemirovski's method \cite{Juditsky2009} is that all outcome probabilities are non-zero. For this reason, we include a small positive parameter $\epsZero \ll 1$ to make the outcome probabilities positive, i.e.,
\begin{equation*}
    p_\sigma^{(l)}(k) = \frac{\tr(E^{(l)}_k \sigma) + \epsZero/N_l}{1 + \epsZero}.
\end{equation*}
We choose $\epsZero = 10^{-5}$ in the numerical simulations so that the Born probabilities are practically unaffected.
Because $\epsZero$ is small with respect to typical experimental imperfections, we will usually drop terms proportional to $\epsZero$ in the theoretical calculations.

Our goal is to construct an estimator $\widehat{F}$ for the fidelity $F(\rho, \sigma)$ between the experimental state $\sigma$ and the target state $\rho$, using the observed outcomes corresponding to the chosen measurement settings. By an estimator, we mean any function that takes the observed outcomes as input and gives an estimate for the fidelity as an output. Constructing a good estimator first requires a measure of error to judge the performance of the estimator. For this purpose, we use the $\failure$-risk defined by Juditsky \& Nemirovski \cite{Juditsky2009}. Intuitively speaking, having an $\failure$-risk of $\adderr$ means that the error in the estimation is at most $\adderr$ with a probability larger than $1 - \failure$. Therefore, the $\failure$-risk defines a symmetric confidence interval for a confidence level of $1 - \failure$. More precisely, given a fidelity estimator $\widehat{F}$ and a confidence level $1 - \failure$, the $\failure$-risk of $\widehat{F}$ is given by
\begin{align}
    \mathcal{R}(\widehat{F}; \failure) = \inf \bigg\{&\adderr'\ \big\vert \sup_{\chi \in \mathcal{X}} \nonumber \\ \Prob_{\text{outcomes} \sim p_\chi}\bigg[ 
    &\left|\widehat{F}(\text{outcomes}) - \tr(\rho \chi)\right| > \adderr' \bigg] < \failure\bigg\}.
    \label{eqn:eps_risk}
\end{align}
Here, ``$\text{outcomes} \sim p_\chi$'' means that the outcomes for $l^\text{th}$ measurement are given by the ($\epsZero$-modified) 
Born rule probabilities $p_\chi^{(l)}$ for a state $\chi$, for all $l = 1, \dotsc, L$. The definition of the $\failure$-risk says that $\mathcal{R}(\widehat{F}; \failure)$ is the smallest number such that the probability that the estimator $\widehat{F}$ has an error larger than $\mathcal{R}(\widehat{F}; \failure)$ is less than $\failure$, irrespective of the underlying state of the system.

Importantly, the $\failure$-risk for any given estimator $\widehat{F}$ can be pre-computed, even before performing a measurement. This is possible because it depends only on the target state, the confidence level, and the chosen measurement settings, while implicitly accounting for all possible measurement outcomes in all possible states. We can thus construct estimators that nearly achieve the minimum possible risk before taking any data.

This minimum possible risk $\mathcal{R}_*(\failure)$ for a chosen confidence level is called the \textit{minimax optimal risk}, and is obtained by minimizing $\mathcal{R}(\widehat{F}; \failure)$ over all possible estimators $\widehat{F}$, i.e., $\mathcal{R}_*(\failure) = \inf_{\widehat{F}} \mathcal{R}(\mathcal{F}; \failure)$. Therefore, the minimax optimal risk gives the smallest possible error by minimizing over all the estimators $\widehat{F}$, while maximizing over all the states $\chi \in \mathcal{X}$ (see Eq.~\eqref{eqn:eps_risk}).

In practice, however, it is computationally very difficult to sift through all possible estimators and choose the optimal one. For this purpose, Juditsky \& Nemirovski~\cite{Juditsky2009} restrict their attention to a subset $\mathcal{F}$ of possible estimators called affine estimators. Any affine estimator $\phi \in \mathcal{F}$ is of the form $\phi = \sum_{l = 1}^L \phi^{(l)}$, where $\phi^{(l)}$ is an estimator for the $l^{\text{th}}$ measurement setting. That is, $\phi^{(l)}$ takes an outcome of the $l^\text{th}$ POVM as input and returns a number as the output. Remarkably, Juditsky \& Nemirovski~\cite{Juditsky2009} show how to construct an affine estimator that achieves nearly minimax optimal performance, and therefore restricting attention to affine estimators is not a problem. Specifically, if $\widehat{F}_* \in \mathcal{F}$ is the near-optimal affine estimator constructed by Juditsky \& Nemirovski's procedure, then the risk $\jnrisk(\failure)$ of $\widehat{F}_*$ computed by their procedure satisfies
\begin{align}
    \mathcal{R}(\widehat{F}_*; \failure) &\leq \jnrisk(\failure) \leq \vartheta(\failure) \mathcal{R}_*(\failure), \label{eqn:risk_minimax_guarantee} \\
    \vartheta(\failure) &= \frac{2 \ln(2/\failure)}{\ln(1/(4\failure))} = 2 + \frac{\ln(64)}{\ln(0.25/\failure)} \label{eqn:risk_minimax_guarantee_factor}
\end{align}
for $\failure \in (0, 0.25)$. We restrict $\failure$ to the interval $(0, 0.25)$ (or equivalently, to confidence levels greater than $75\%$) so that $\vartheta(\failure)$ is well-defined. From Eq.~\eqref{eqn:risk_minimax_guarantee}, we see that the risk $\jnrisk(\failure)$ computed by Juditsky \& Nemirovski's procedure is actually an upper bound on the (unknown) $\failure$-risk $\mathcal{R}(\widehat{F}_*, \failure)$ of the near-optimal affine estimator $\widehat{F}_*$. Since the probability that the estimator $\widehat{F}_*$ fails by an error more than $\mathcal{R}(\widehat{F}_*, \failure)$ is less than $\failure$, and because $\mathcal{R}(\widehat{F}_*, \failure) \leq \jnrisk(\failure)$ we can conclude that the probability that $\widehat{F}_*$ fails by more than $\jnrisk(\failure)$ is also less than $\failure$. Therefore, the risk $\jnrisk(\failure)$ defines a confidence interval corresponding to the chosen confidence level $1 - \failure$.

Eq.~\eqref{eqn:risk_minimax_guarantee} guarantees that the confidence interval defined by $\jnrisk(\failure)$ is nearly minimax optimal. This is because the risk $\jnrisk(\failure)$ of the near-optimal affine estimator $\widehat{F}_*$ computed by Juditsky \& Nemirovski's procedure is at most a constant times the minimax optimal risk, where the constant factor $\vartheta(\failure)$ depends only on the chosen confidence level. Note that $\vartheta(0.1) < 6.54$ and that $\vartheta(\failure)$ is a decreasing function of $\failure$, converging to $2$ as $\failure \to 0$. Therefore, for confidence levels greater than $90\%$, $\jnrisk$ is guaranteed to be close to the minimax optimal risk by a factor less than $6.5$. In practice, the estimator typically performs better than the theoretically guaranteed bound of Eq.~\eqref{eqn:risk_minimax_guarantee}. We use the star symbol as a subscript in $\jnrisk(\failure)$ to emphasize that it is nearly minimax optimal.

In the remainder of this study, we omit the argument $\failure$ when writing $\jnrisk$ for the sake of notational simplicity. A short summary of the different types of risk mentioned in this section can be found in Table~\ref{tab:types_of_risk}.

\begin{table}[h]
    \begin{tabular}{l p{7.3cm}}
        $\mathcal{R}_*(\failure)$ & Minimax optimal risk, which is the minimum possible additive error with a confidence level of $1 - \failure$, accounting for all possible states, measurement outcomes, and all possible estimators. \\[0.3cm]
        $\mathcal{R}(\widehat{F}_*; \failure)$ & $\failure$-risk of the estimator $\widehat{F}_*$ as defined in Eq.~\eqref{eqn:eps_risk}. This is the smallest possible additive error for the estimator $\widehat{F}_*$ with a confidence level of $1 - \failure$, accounting for all possible states and measurement outcomes. \\[0.3cm]
        $\jnrisk(\failure)$ & Risk (or additive error) of the estimator $\widehat{F}_*$ for a confidence level of $1 - \failure$ that is computed by Juditsky \& Nemirovski's procedure. This is an upper bound on the $\failure$-risk $\mathcal{R}(\widehat{F}_*; \failure)$ and it is nearly minimax optimal in the sense of Eq.~\eqref{eqn:risk_minimax_guarantee}.
    \end{tabular}
    \caption{Different types of risk considered in this study. In each case, the true state is unknown, the measurement settings are fixed, and all measurement outcomes consistent with the state and measurement settings are allowed. $\widehat{F}_*$ is the near-optimal affine estimator constructed by Juditsky \& Nemirovski's \cite{Juditsky2009} procedure.}
    \label{tab:types_of_risk}
\end{table}

A practically implementable version~\cite{seshadri2021computation} of the procedure outlined by Juditsky \& Nemirovski \cite{Juditsky2009} for constructing the near-optimal affine estimator $\widehat{F}_*$ and its associated risk $\jnrisk$ is as follows:
\begin{enumerate}[leftmargin=0.2cm]
    \item Find the saddle-point value of the function $\Phi\colon (\mathcal{X} \times \mathcal{X}) \times (\mathcal{F} \times \mathbb{R}_+) \to \mathbb{R}$ defined as
    \begin{align}
        \Phi(\chi_1, &\chi_2;\ \phi, \alpha) = \tr(\rho \chi_1) - \tr(\rho \chi_2) + 2\alpha \ln(2/\failure) \nonumber \\
            &+ \alpha \sum_{l = 1}^L R_l \Bigg[\ln\left(\sum_{k = 1}^{N_l} e^{-\phi^{(l)}_k/\alpha} \frac{\tr(E^{(l)}_k \chi_1) + \epsZero/N_l}{1 + \epsZero}\right) \nonumber \\
            &\hspace{1cm} + \ln\left(\sum_{k = 1}^{N_l} e^{\phi^{(l)}_k/\alpha} \frac{\tr(E^{(l)}_k \chi_2) + \epsZero/N_l}{1 + \epsZero}\right)\Bigg] \label{eqn:Phi_quantum}
    \end{align}
    to a given precision, where $\rho$ is the target state, $\chi_1, \chi_2 \in \mathcal{X}$ are density matrices, $\phi \in \mathcal{F}$ is an affine estimator, and $\alpha > 0$ is a positive number. The number $\alpha$ as such has no intuitive meaning, and enters the function $\Phi$ for mathematical reasons. The function $\Phi$ is itself a mathematical device designed by Juditsky \& Nemirovski to compute the near-optimal affine estimator and the associated risk.
    
    Recall that any affine estimator $\phi \in \mathcal{F}$ can be written as $\phi = \sum_{l = 1}^L \phi^{(l)}$. Here, $\phi^{(l)}$ is a function taking outcomes of the $l^{\text{th}}$ POVM as input, and giving a real number as an output. The $l^{\text{th}}$ POVM has $N_l$ measurement outcomes, and $\phi_k^{(l)}$ denotes the output of $\phi^{(l)}$ for the $k^{\text{th}}$ measurement outcome, for each $k = 1, \dotsc, N_l$.
    
    Juditsky \& Nemirvoski \cite{Juditsky2009} prove that the function $\Phi$ has a well-defined saddle-point value (see Appendix~\ref{app:minimax_theory} for more details). We present an algorithm to compute the saddle-point value of $\Phi$ to a given precision using convex optimization in Appendix~\ref{app:minimax_numerical}.

    \item We denote the saddle-point value of $\Phi$ by $2\jnrisk$, i.e.,
        \begin{align}
            \jnrisk &= \frac{1}{2} \sup_{\chi_1, \chi_2 \in \mathcal{X}} \inf_{\phi \in \mathcal{F}, \alpha > 0} \Phi(\chi_1, \chi_2; \phi, \alpha) \nonumber \\
                                    &= \frac{1}{2} \inf_{\phi \in \mathcal{F}, \alpha > 0} \max_{\chi_1, \chi_2 \in \mathcal{X}} \Phi(\chi_1, \chi_2; \phi, \alpha) . \label{eqn:JN_risk_saddle_point}
        \end{align}
        Suppose that the saddle-point value is attained at $\chi_1^*, \chi_2^* \in \mathcal{X}$, $\phi_* \in \mathcal{F}$ and $\alpha_* > 0$ to within the given precision. Since $\phi^* \in \mathcal{F}$ is an affine estimator, we can write $\phi_* = \sum_{l = 1}^L \phi^{(l)}_*$. Suppose that independent and identically distributed outcomes $\{o^{(l)}_1,\dotsc, o^{(l)}_{R_l}\}$ are observed upon measurement of the  $l^\text{th}$ POVM. Then, the optimal estimator $\widehat{F}_* \in \mathcal{F}$ for estimating the fidelity of the state prepared in the lab with the target state is given as
        \begin{align}
            \widehat{F}_*(\{o^{(l)}_1, \dotsc, &o^{(l)}_{R_l}\}_{l = 1}^L) = \sum_{l = 1}^L \sum_{r = 1}^{R_l} \phi^{(l)}_*(o^{(l)}_r) + c , \label{eqn:JN_estimator}
            \intertext{where the constant $c$ is}
            c &= \frac{1}{2} \left(\tr(\rho \chi_1^*) + \tr(\rho \chi_2^*)\right) . \label{eqn:JN_estimator_constant}
        \end{align}
        The procedure outputs $\widehat{F}_*$ and $\jnrisk$, such that $\mathcal{R}(\widehat{F}_*; \failure)~\leq~\jnrisk$.
        Then, from Eq.~\eqref{eqn:eps_risk}, we can infer that $|\widehat{F}_*(\text{outcomes}) - \tr(\rho \sigma)| \leq \jnrisk$ with probability greater than or equal to $1 - \failure$, where $\sigma$ is the actual state prepared in the lab.
        The guarantees given by Eqs.~\eqref{eqn:risk_minimax_guarantee} \& \eqref{eqn:risk_minimax_guarantee_factor} apply.
\end{enumerate}

Observe that $\widehat{F}_*$ is a function of the outcomes (see Eq.~\eqref{eqn:JN_estimator}), while, the risk $\jnrisk$ of the estimator is a function of the measurement protocol. We thus know how well the constructed estimator does even before performing an experiment, which can be used for benchmarking different measurement protocols.

It is not apparent from the above procedure how finding the saddle point of the function $\Phi$ defined in Eq.~\eqref{eqn:Phi_quantum} leads to constructing the fidelity estimator $\widehat{F}_*$ with error $\jnrisk$.
Fortunately, we can get some mathematical intuition on how this algorithm works by rewriting the function $\Phi$ in Eq.~\eqref{eqn:Phi_quantum}, following the proof of Juditsky \& Nemirovski~\cite{Juditsky2009}.
For simplicity, we restrict our attention to the case where we record a single outcome of a single POVM (i.e., $L = 1$ and $R_1 = 1$), and the more general case can be handled similarly.
Denoting the outcome probability as $p^{(1)}_k(\chi) = (\tr(E^{(1)}_k \chi) + \epsZero) / (1 + \epsZero)$, we can rearrange terms in Eq.~\eqref{eqn:Phi_quantum} to obtain
\begin{align*}
    &\frac{\Phi(\chi_1, \chi_2; \phi_*, \alpha_*) - 2 \jnrisk}{\alpha_*} \\
    &= \Bigg[\ln\left(\sum_{k = 1}^{N_1} \exp\left((\tr(\rho \chi_1) - (\phi_{* k}^{(1)} + c) - \jnrisk) / \alpha_*\right) p_k(\chi_1)\right) \\
    &\qquad + \ln(2/\delta)\Bigg] \\
    &+ \Bigg[\ln\left(\sum_{k = 1}^{N_1} \exp\left((-\tr(\rho \chi_2) + (\phi_{* k}^{(1)} + c) - \jnrisk) / \alpha_*\right) p_k(\chi_2)\right) \\
    &\qquad + \ln(2/\delta)\Bigg].
\end{align*}
Since $\max_{\chi_1, \chi_2} \Phi(\chi_1, \chi_2; \phi_*, \alpha_*) = 2 \jnrisk$ to within some precision, we have $\Phi(\chi_1, \chi_2; \phi_*, \alpha_*) - 2 \jnrisk \leq 0$ for all $\chi_1, \chi_2$.
Then, noting that $\widehat{F}_*(k) = \phi_{* k}^{(1)} + c$ (for one measurement setting), the choice of $c$ in Eq.~\eqref{eqn:JN_estimator_constant} ensures that
both the terms in square brackets above are less than or equal to zero~\cite{Juditsky2009}.
In other words, for all $\chi_1, \chi_2$, we have
\begin{align*}
    &\sum_{k = 1}^{N_1} \exp\left((\tr(\rho \chi_1) - \widehat{F}_*(k) - \jnrisk) / \alpha_*\right) p_k(\chi_1) \leq \delta/2 \\
    &\sum_{k = 1}^{N_1} \exp\left((-\tr(\rho \chi_2) + \widehat{F}_*(k) - \jnrisk) / \alpha_*\right) p_k(\chi_2) \leq \delta/2,
\end{align*}
where by $\widehat{F}_*(k)$ we mean $\widehat{F}_*$ evaluated on the $k$th outcome.
However, since $e^x \geq 1$ for $x \geq 0$ and $\delta \in (0, 0.25)$, for the above inequalities to be true, we must have
$\tr(\rho \chi) - \widehat{F}_* - \jnrisk \leq 0$ and $-\tr(\rho \chi) + \widehat{F}_* - \jnrisk \leq 0$ with high probability for any state~$\chi$.
This is equivalent to saying $|\tr(\rho \chi) - \widehat{F}_*| \leq \jnrisk$ with high probability for all states~$\chi$.
Note that this idea can be formalized using Markov's inequality.
Therefore, $\widehat{F}_*$ gives a valid estimate for the fidelity $\tr(\rho \chi)$ within error $\jnrisk$ with high probability, no matter what the underlying state $\chi$ is.

Although the formal procedure above is rather abstract, in section~\ref{secn:minimax_method_affine_estimator}, we show that the estimator $\widehat{F}_*$ can be understood as an appropriately weighted sum of the observed frequencies, where the weights depend on the target state and measurement scheme. In section~\ref{secn:minimax_method_optimal_risk}, we show that the risk $\jnrisk$ can be written in terms of classical fidelities determined by the POVMs in the measurement protocol. This is helpful in performing theoretical calculations involving the risk, and potentially useful for benchmarking measurement protocols.

Importantly, the above procedure can not only be used to construct an estimator for the fidelity, but for the expectation value of any observable. This can be achieved by replacing the target state $\rho$ in the above equations with the Hermitian operator $\mathcal{O}$ corresponding to the observable whose expectation value we wish to estimate.

We discuss the theoretical underpinnings of the minimax method in Appendix~\ref{app:minimax_theory}, starting with a brief overview of Juditsky \& Nemirovski's general set-up \cite{Juditsky2009}, followed by details on adapting their method to estimating fidelity and expectation values. Subsequently, details regarding the numerical implementation of the minimax method are given in Appendix~\ref{app:minimax_numerical}. In particular, we outline a convex optimization algorithm for constructing the optimal estimator $\widehat{F}_*$ and the associated risk $\jnrisk$ given any target state and measurement settings.

\subsection{Toy problem: 1-qubit target state}
We explain in detail a simple setting where it is possible to understand the estimator that is constructed by the procedure described above. Let the target state be $\rho = \op{1}{1}$, where $\{\ket{0}, \ket{1}\}$ are the eigenstates of Pauli $Z$. Suppose that our experiment consists of performing Pauli $Z$ measurements, and we want to estimate the fidelity of the lab state $\sigma$ with the target state $\rho$. This problem is essentially classical: given the quantum state $\sigma$, we can write it in the computational basis as $\sigma = (1 - p) \op{0}{0} + q \op{0}{1} + q^* \op{1}{0} + p \op{1}{1}$. Since we measure $Z$, we will obtain the outcome $\ket{0}$ with a probability of $1 - p$ and $\ket{1}$ with a probability of $p$, and the problem is to estimate the fidelity $\tr(\rho \sigma) = p$. This is the same as having a Bernoulli random variable with parameter $p$, which we want to estimate. Note that the parameter $p$ corresponds to the probability of observing the outcome $1$ for the Bernoulli random variable.

In order to numerically compute the fidelity estimator, we choose a confidence level of $95\%$ and $R = 100$ repetitions of Pauli $Z$ measurement. Since we measure only one POVM, the estimator given in Eq.~\eqref{eqn:JN_estimator} can be written as $\widehat{F}_*(\{o_1, \dotsc, o_R\}) = \sum_{i = 1}^R \phi_*(o_i) + c$, where $o_1, \dotsc, o_R \in \{0, 1\}$ are the measurement outcomes. Here, $\phi_* \in \mathcal{F}$ is an affine estimator corresponding to the saddle-point of the function $\Phi$ (see Eq.~\eqref{eqn:Phi_quantum}) that accepts a measurement outcome ($0$ or $1$) as input and gives a number as an output, while the constant $c$ is as given in Eq.~\eqref{eqn:JN_estimator_constant}. Numerically, we find that the affine estimator $\phi_*$ gives the values $\phi_*(0) \approx -0.476 \times 10^{-2}$ and $\phi_*(1) \approx 0.476 \times 10^{-2}$ for the measurement outcomes $0$ and $1$, respectively. The value of the constant $c$ computed numerically is $c = 0.5$.

Now, we show that the estimator $\widehat{F}_*$ constructed by the minimax method is essentially the sample mean estimator for a Bernoulli random variable, but with a small additive constant that pushes the estimate away from the boundary of the parameter space. Note that the sample mean estimator is also the Maximum Likelihood estimator for estimating the parameter of a Bernoulli random variable. Recall that our estimator for fidelity is given as $\widehat{F}_*(\{o_1, \dotsc, o_R\}) = \sum_{i = 1}^R \phi_*(o_i) + c$ for a given set of outcomes $\{o_i\}_{i = 1}^R$. We can then write the affine estimator $\phi_*$ as $\phi_*(x) = (\phi_*(1) - \phi_*(0)) x + \phi_*(0)$ for any outcome $x \in \{0, 1\}$. Therefore,
\begin{align}
    \widehat{F}_*(\{o_i\}_{i = 1}^R) &= \sum_{i = 1}^R \left((\phi_*(1) - \phi_*(0)) o_i + \phi_*(0)\right) + c \notag \\
                                     &= (\phi_*(1) - \phi_*(0)) \sum_{i = 1}^R o_i + (R \phi_*(0) + c) \notag
    \intertext{where $o_i \in \{0, 1\}$. Using the values obtained numerically, we find that}
    \widehat{F}_*(\{o_i\}_{i = 1}^{100}) &= \frac{0.952}{100} \sum_{i = 1}^{100} o_i + 0.024 \nonumber
\end{align}
Thus, we can interpret $\widehat{F}_*$ as (approximately) the mean of the sample $o_1, \dotsc, o_{100}$, but with a small additive constant of $0.024$. For finite sample sizes the constant term is justified, because even upon observing all $0$ outcomes, we cannot be certain that the probability $p$ of observing outcome $1$ is $p = 0$. For a similar reason, we have $0.952/100$ as the coefficient for the sum instead of $1/100$. We observe numerically that this coefficient approaches $1 / R$ and the additive constant goes to $0$ as $R$ increases.

What if we decide to measure Pauli $X$ (measurement~$2$) in addition to Pauli $Z$ (measurement~$1$)? In this case, the estimator $\widehat{F}_*$ can be written as $\widehat{F}_*(\{o_1, \dotsc, o_R\}) = \sum_{i = 1}^R \phi^{(1)}_*(o_i) + \sum_{i = 1}^R \phi^{(2)}_*(o_i) + c$. Here, $\phi^{(1)}_*$ is the affine estimator at the saddle-point corresponding to the Pauli $Z$ measurement, while $\phi^{(2)}_*$ is the affine estimator at the saddle-point corresponding to the Pauli $X$ measurement. In such a scenario, we find that the saddle-point value corresponding to the Pauli $Z$ measurement is unchanged. That is, $\phi^{(1)}_*(0) \approx -0.476 \times 10^{-2},\ \phi^{(1)}_*(1) \approx 0.476 \times 10^{-2}$, and the constant $c = 0.5$ is also the same as before. On the other hand, the saddle-point value corresponding to Pauli $X$ measurement is $\phi^{(2)}_*(0) = \phi^{(2)}_*(1) = 0$. Therefore, the outcomes from the $Z$ measurement are weighted as before, but those from $X$ measurement are discarded. The reason is simple: for estimating the fidelity with $\rho = \op{1}{1}$, measurement of Pauli $X$ gives no useful information, irrespective of what the actual state $\sigma$ is. Indeed, upon measuring $X$, we get $\ket{+}$ with probability $1/2 + \Re(q)$ and $\ket{-}$ with probability $1/2 - \Re(q)$, both of which are independent of $\tr(\rho \sigma) = p$ that we want to estimate. Thus, our method properly incorporates the available measurements to give an estimator for fidelity.

\subsection{Affine estimator for fidelity\label{secn:minimax_method_affine_estimator}}
An affine function is a linear function plus a constant term. In the previous section, we saw that the estimator for the $1$-qubit toy problem can be expressed as an affine function of the sample mean. In this section, we show that such a result holds more generally. Our estimator can be written as an affine function of the observed frequencies.

To begin, recall that the outcomes of the $l$th measurement setting is described by the set $\Omega^{(l)} = \{1, \dotsc, N_l\}$, where the index $k$ corresponds to the POVM element $E^{(l)}_k$. Corresponding to each outcome $k \in \Omega^{(l)}$, we associate a canonical basis element $\bm{e}^{(l)}_k \in \mathbb{R}^{N_l}$. From Eq.~\eqref{eqn:JN_estimator}, we know that the fidelity estimator is given as
\begin{equation*}
    \widehat{F}_*(\{o^{(l)}_1, \dotsc, o^{(l)}_{R_l}\}_{l = 1}^{R_l}) = \sum_{l = 1}^L \sum_{r = 1}^{R_l} \phi^{(l)}_*(o^{(l)}_r) + c
\end{equation*}
where $\{o^{(l)}_1, \dotsc, o^{(l)}_{R_l}\}$ are the outcomes corresponding to the $l^\text{th}$ measurement setting, and the constant $c$ is given by Eq.~\eqref{eqn:JN_estimator_constant}. In the form written above, $\phi^{(l)}_* \in \mathcal{F}^{(l)}$ are real-valued functions defined on the set $\Omega^{(l)}$. Regarding $\phi^{(l)}_*$ as an $N_l$-dimensional real vector with ``coefficient'' vector
\begin{equation*}
    \bm{a}^{(l)} = \begin{pmatrix} \phi^{(l)}_*(1), & \dotsc, & \phi^{(l)}_*(N_l) \end{pmatrix} \in \mathbb{R}^{N_l}
\end{equation*}
we can write $\phi^{(l)}_*(k) = \langle\bm{a}^{(l)}, \bm{e}^{(l)}_k\rangle$ for ${k\in\Omega^{(l)}}$. Then, the estimator can be written as
\begin{equation*}
    \widehat{F}_*(\{o^{(l)}_1, \dotsc, o^{(l)}_{R_l}\}_{l = 1}^{R_l}) = \sum_{l = 1}^L \sum_{r = 1}^{R_l} \ip{\bm{a}^{(l)}, \bm{e}^{(l)}_{o^{(l)}_r}} + c
\end{equation*}
which is affine in $\bm{e}^{(l)}_{o^{(l)}_r}$. However, in the above equation, $\widehat{F}_*$ is not an affine function of the input. To remedy this, we define the vector of experimentally observed frequencies $\bm{f}^{(l)} \in \mathbb{R}^{N_l}$, obtained by binning the outcomes $o^{(l)}_1, \dotsc, o^{(l)}_{R_l}$, such that $f^{(l)}_k$ denotes the relative frequency of the outcome $k \in \Omega^{(l)}$ observed in the experiment for the $l$th measurement setting. We can then write
\begin{equation}
    \bm{f}^{(l)} = \frac{1}{R_l} \sum_{r = 1}^{R_l} \bm{e}^{(l)}_{o^{(l)}_r} . \notag
\end{equation}
We can then express the fidelity estimator in a way that is indeed affine in the observed frequencies as
\begin{equation}
    \widehat{F}_*(\bm{f}^{(1)}, \dotsc, \bm{f}^{(L)}) = \sum_{l = 1}^L R_l \ip{\bm{a}^{(l)}, \bm{f}^{(l)}} + c. \label{eqn:JN_estimator_affine}
\end{equation}
Eq.~\eqref{eqn:JN_estimator_affine} gives an intuitive understanding of the estimator constructed by the minimax method --- it is simply an appropriate weighting of the relative frequencies observed in the experiment. The weights are obtained from the saddle-point of the function $\Phi$ defined in Eq.~\eqref{eqn:Phi_quantum}.

\subsection{Risk and the best sample complexity\label{secn:minimax_method_optimal_risk}}
We now learned about the estimator, but what about the risk? We found that the risk $\jnrisk$ given by Eq.~\eqref{eqn:JN_risk_saddle_point} is half the saddle-point value of the function $\Phi$. However, in this form, it is difficult to infer what this quantity is. In Appendix~\ref{app:minimax_theory}, we show that one can write the risk $\jnrisk$ given by the minimax method in a form that is more amenable to interpretation as
\begin{align}
    &\jnrisk = \max_{\chi_1, \chi_2 \in \mathcal{X}} \frac{1}{2} \left(\tr(\rho \chi_1) - \tr(\rho \chi_2)\right) \nonumber \\
    &\hspace{1.6cm} \text{s.t.}\ \prod_{l = 1}^L \left[F_C(\chi_1, \chi_2, \{E^{(l)}_k\})\right]^{R_l/2} \geq \frac{\failure}{2} \label{eqn:JNriskprop3.1}
    \intertext{where}
    &F_C(\chi_1, \chi_2, \{E^{(l)}_k\}) = \left(\sum_{k = 1}^{N_l} \sqrt{\tr\left(E^{(l)}_k \chi_1\right) \tr\left(E^{(l)}_k \chi_2\right)}\right)^2 \label{eqn:classicalfidelity}
\end{align}
denotes the classical fidelity between the probabilities determined by states $\chi_1$ and $\chi_2$ corresponding to the POVM measurement $\{E^{(l)}_k\}$. Here, the measurement protocol chosen by the experimenter corresponds to measuring $L$ different POVMs, where the $l^{\text{th}}$ POVM $\{E^{(l)}_1, \dotsc, E^{(l)}_{N_l}\}$ is measured $R_l$ times, for $l = 1, \dotsc, L$. Thus, the risk $\jnrisk$ can be intuitively interpreted as half the largest deviation $F(\rho, \chi_1) - F(\rho, \chi_2)$ of the fidelity of the target state $\rho$ with the states $\chi_1, \chi_2$ compatible with the measurement protocol and confidence level. The factor of half comes from the fact that we are looking for a symmetric confidence interval.

Note that the fidelity between any two quantum states must lie between $0$ and $1$, and thus, $0 \leq \tr(\rho \chi) \leq 1$ for any density matrix $\chi$. Then, we can infer from Eq.~\eqref{eqn:JNriskprop3.1} that the maximum possible value of risk is $0.5$. Indeed, an uncertainty of $\pm 0.5$ gives a confidence interval of length $1$, which is also the size of interval for fidelity.

On the other hand, it can be shown from Eq.~\eqref{eqn:JNriskprop3.1} that the risk decreases when the number of measurement settings or the number of shots for any measurement setting is increased. To that end, note that $F_C$ ranges between $0$ and $1$, so raising it to a larger power makes it smaller. Thus, increasing the number of shots $R_l$ makes it harder to satisfy the constraint that $\prod_l F_C^{R_l/2}$ must be larger than $\failure/2$. Similarly, if we include another measurement setting (in effect, increasing $L$), we have one more fraction $F_C$ multiplying the left-hand side of the constraint, thereby making the constraint tighter. Since a tighter constraint implies a more restricted search space for maximization, we can infer that the risk will become smaller (or stay the same) in either case. This also shows that the risk is dependent on the chosen measurement protocol, and protocols that make the constraint tighter will have a smaller risk in estimating the fidelity.

We numerically quantify the variation of the risk with the number of measurement settings and the number of repetitions for a $4$-qubit randomly chosen target state. We apply $10\%$ depolarizing noise to obtain the actual state, and perform Pauli measurements. We see from Fig.~\ref{fig:JN_risk_analysis} that the risk decreases with the number of Pauli measurements as well as the number of repetitions.
\begin{figure}[!ht]
    \includegraphics[width=0.5\textwidth]{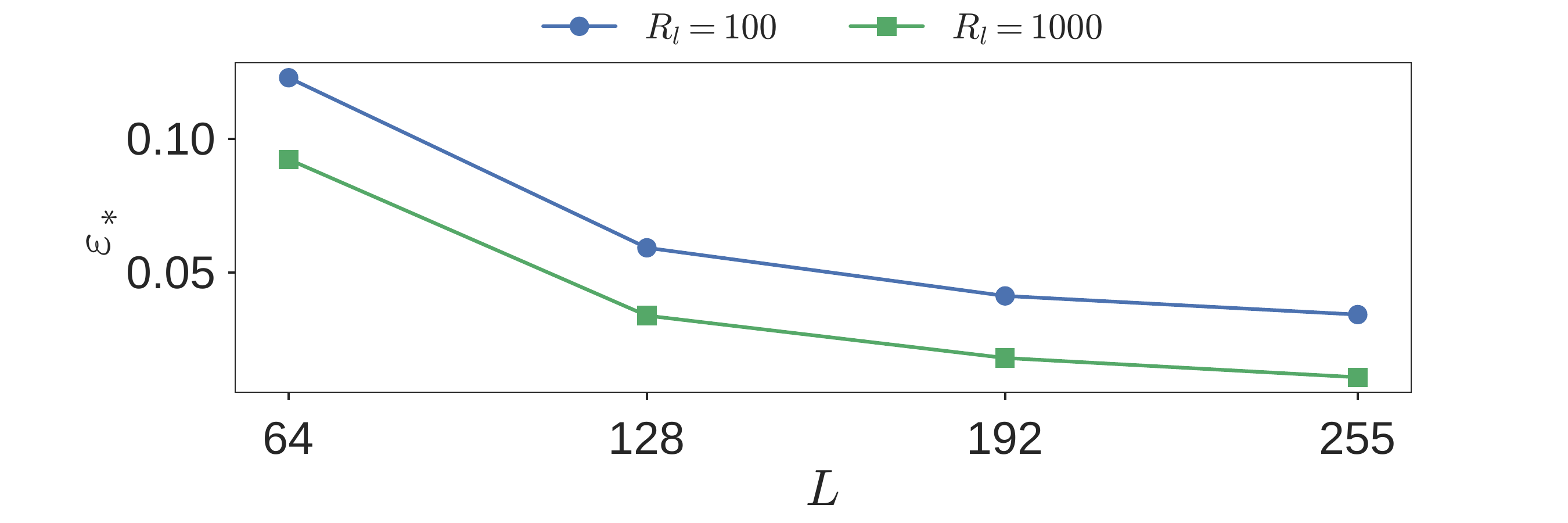}
    \caption{Variation of the risk with the number of Pauli measurements $L$ and number of repetitions $R_L$ of each Pauli measurement for a $4$-qubit random target state. All the computed risks correspond to a $95\%$ confidence level.}
    \label{fig:JN_risk_analysis}
\end{figure}

A natural question that arises is how low the risk can be. To find a lower bound, recall that we can write the fidelity between any two states as a minimization of the classical fidelity over all POVMs \cite{fuchs1996distinguishability}.
\begin{equation*}
    F(\chi_1, \chi_2) = \min_{\text{POVM } \{F_i\}} F_C(\chi_1, \chi_2, \{F_i\})
\end{equation*}
In particular, we have $F(\chi_1, \chi_2) \leq F_C(\chi_1, \chi_2, \{E^{(l)}_k\})$ for every POVM $\{E^{(l)}_k\}$ that we are using. Thus, we obtain the following lower bound on our risk
\begin{align*}
    \jnrisk \geq &\max_{\chi_1, \chi_2 \in \mathcal{X}} \frac{1}{2} \left(\tr(\rho \chi_1) - \tr(\rho \chi_2)\right) \nonumber \\
                                 &\qquad \text{s.t.}\quad F(\chi_1, \chi_2) \geq \left(\frac{\failure}{2}\right)^{\frac{2}{R}}
\end{align*}
where $R = \sum_{l = 1}^L R_l$ is the total number of shots. Evaluating the right-hand side of the above equation, we obtain a lower bound for the risk. This result is summarized in the following theorem.
\begin{theorem}
    \label{thm:minimax_method_best_sample_complexity}
    Let $\rho$ be any pure target state. Suppose that $\jnrisk$ is the risk associated with the fidelity estimator given by the minimax method corresponding to a confidence level of $1 - \failure \in (0.75, 1)$ and any measurement scheme. Then, the risk is bounded below as
    \begin{equation}
        \jnrisk \geq \frac{1}{2} \sqrt{1 - \left(\frac{\failure}{2}\right)^{2/R}}  \label{eqn:optimalJNrisk}
    \end{equation}
    where $R$ is the total number of measurement outcomes \textnormal{(}all measurement settings combined\textnormal{)} to be supplied to the estimator.
    This lower bound can be achieved by measuring in the basis defined by $\rho$. That is, $R$ repetitions of the POVM $\{\rho, \id - \rho\}$ achieves the risk $\jnrisk$.

    Stated differently, the best sample complexity that can be obtained using the minimax method corresponding to a risk of $\jnrisk \in (0, 0.5)$ and confidence level $1 - \failure$ is
    given by
    \begin{align}
        R &\geq \frac{2 \ln(2/\failure)}{\left|\ln(1 - 4\jnrisk^2)\right|}  \label{eq:riskBound} \\
          &\approx \frac{\ln(2/\failure)}{2\jnrisk^2} \text{ when } \jnrisk^2\ll 1. \nonumber
    \end{align}
\end{theorem}
\begin{proof}
    See Appendix~\ref{proof:minimax_method_best_sample_complexity}.
\end{proof}
Note that sample complexity refers to the smallest number of measurements outcomes required for achieving the desired risk for the chosen confidence level.

Observe that the lower bound on the risk in Eq.~\eqref{eqn:optimalJNrisk} is independent of the system dimension, the target state and the true state. It only depends on the confidence level and total number of repetitions.
Note that it is possible to obtain dimension-independent sample complexity because we assume that the target state $\rho$ is pure.
Without this assumption, the sample complexity for fidelity estimation is expected to scale with the dimension (see, for example, Ref.~\cite{buadescu2019quantum, haah2016sample}).
That said, the assumption that the target state is pure is well-motivated from a practical standpoint as we usually seek to prepare pure states in experiments (e.g., Bell states or GHZ states).
Indeed, this is a common assumption in many fidelity estimation studies~\cite{Flammia2011, DaSilva2011, Huang2020}.
Since the true state is allowed to be mixed, our method is applicable in practical scenarios.
Assuming that the target state is pure, one can, in principle, achieve the optimal sample complexity by measuring in a basis defined by the target state $\rho$.
We acknowledge that, in practice, it is often impractical to implement the POVM $\{\rho, \id - \rho\}$ (or an orthonormal basis containing $\rho$) that achieves this sample complexity.
However, next we show that we can achieve something very close to this optimal sample complexity for stabilizer states using a practical measurement scheme.

\subsection{Stabilizer states\label{secn:minimax_method_stabilizer_states}}
Stabilizer states are the cornerstone of numerous applications, ranging from measurement-based quantum computing~\cite{BriegelMBQC2009} to quantum error correction~\cite{Gottesman1997}. An $n$-qubit stabilizer state $\rho$ is the unique $+1$ eigenstate of exactly $n$ Pauli operators that generate a stabilizer group of size $d = 2^n$. The state $\rho$ can thus be written as \cite{klieschcharacterization}
\begin{equation}
    \rho = \frac{1}{d} \left(\id + \sum_{S \in \mathcal{S}_{n} \setminus \{\id\}} S\right)
\end{equation}
where $\mathcal{S}_n$ is the stabilizer subgroup corresponding to the stabilizer state $\rho$.

To estimate the fidelity with a target stabilizer state efficiently, we implement the minimax optimal strategy \cite{pallister2018optimal, klieschcharacterization} while restricting to Pauli measurements. 
The measurement strategy is simple: uniformly sample an element from the stabilizer subgroup (excluding the identity), and record the eigenvalue of the outcome (whether $+1$ or $-1$) of the stabilizer measurement. This strategy is very similar to DFE, with the exception that we exclude the identity. This measurement scheme can be implemented by an effective POVM with elements $\{\Theta,\Delta_\Theta\}$ given by
\begin{align}
    \Theta &= \rho + \frac{d/2 - 1}{d - 1} \Delta_\rho \nonumber \\
    \Delta_\Theta &= \frac{d/2}{d - 1} \Delta_\rho , \label{eqn:stabilizer_minimaxoptimal_POVM}
\end{align}
where $\Delta_\rho = \id - \rho$. See the proof of Theorem \ref{thm:minimax_method_pauli_scheme_sample_complexity} for details on how this is obtained (also see Ref.~\cite{pallister2018optimal, klieschcharacterization}). We can see that the effective POVM is a combination of the target state $\rho$ and $\Delta_\rho$ (which is orthogonal to $\rho$). This is similar to measuring in the basis defined by the target state, so from our results in section~\ref{secn:minimax_method_optimal_risk}, we can expect the risk to be small. Indeed, we find that for a sufficient number of repetitions of this measurement (independent of the dimension or the stabilizer state), the risk of the estimator is at most four times the optimal risk described in Eq.~\eqref{eqn:optimalJNrisk}. We summarize this result below.
\begin{proposition}
    \label{prop:minimax_method_stabilizer_sample_complexity}
    Let $\rho$ be an $n$-qubit stabilizer state. Suppose that we uniformly sample $R$ elements from the stabilizer group (with replacement) and measure them. Then, for any risk $\jnrisk \in (0, 0.5)$ and confidence level ${1-\failure} \in (0.75, 1)$,
    \begin{align}
        R &\geq 2 \frac{\ln\left(2/\failure\right)}{\left|\ln\left(1 - \left(\frac{d}{d - 1}\right)^2 \jnrisk^2\right)\right|} \label{eqn:sample_complexity_stabilizers} \\
          &\approx 2 \left(\frac{d - 1}{d}\right)^2 \frac{\ln(2/\failure)}{\jnrisk^2}
          \text{ when } \jnrisk^2\ll 1
    \end{align}
    suffices to build an estimator using the minimax method that achieves this risk. Here, $d = 2^n$ is the dimension of the system.
\end{proposition}
\begin{proof}
    Take $\xi = d$ in corollary \ref{corr:minimax_sample_complexity_stabilizer} in the appendix.
\end{proof}

Since $(d - 1)/d < 1$, the dimension dependence actually improves the sample complexity, and for large dimensions one needs to measure only a constant number of stabilizers. This sample complexity is of the same order as that obtained using DFE.

Furthermore, following the ideas given in the proof of Theorem \ref{thm:minimax_sample_complexity_2outcomePOVM}, we can simplify the algorithm given in Appendix \ref{app:minimax_numerical} to compute the estimator. In the simplified algorithm, we only need to perform a two-dimensional optimization regardless of the dimension of the system. Consequently, building the estimator is both time and memory efficient.

In the following we present simulated results for $2$, $3$ and $4$ qubit stabilizer states, for a risk of $0.05$ and a confidence level of $95\%$. We implement the measurement strategy discussed above, namely randomly sampling $R$ stabilizers and recording the outcome eigenvalue. The estimated fidelity in each case is summarized in Table~\ref{tab:JNstabilizers}. Even for the $4$-qubit stabilizer state, constructing the estimator and generating measurements to compute the estimate together took only a few seconds on a personal computer.
\begin{table}[!ht]
    \centering
        \begin{tabular}{l c c c c}
        \toprule
        & & \multicolumn{3}{c}{$n$} \\
        \cmidrule(l){3-5}
            &  & 2 & 3 & 4 \\
            \midrule
            True fidelity & $F(\rho, \sigma)$ & 0.925 & 0.912 & 0.906 \\
            Estimate of fidelity & $\widehat{F}_*(\text{outcomes})$ & 0.933 & 0.911 & 0.898 \\
            Number samples & $R$ & 1657 & 2256 & 2591 \\
            \bottomrule
        \end{tabular}
    \caption{True fidelity, the fidelity estimate, and the number of stabilizers sampled corresponding to a risk of $0.05$ and $95\%$ confidence level for $n$-qubit stabilizer states.}
    \label{tab:JNstabilizers}
\end{table}

Now, on the other extreme, we consider the case where an insufficient number of stabilizer measurements are provided. Specifically, we consider a measurement protocol where we measure only $n - 1$ generators of an $n$-qubit stabilizer state. We assume these to be subspace measurements, i.e., the measurements correspond to projecting on the eigenspaces of the generators with eigenvalue $\pm 1$. This is an insufficient measurement protocol because the measurement of $n - 1$ stabilizers cannot uniquely identify the stabilizer state, as there are two orthogonal states that will be consistent with the measurement statistics. As a consequence, any reasonable fidelity estimation protocol should give complete uncertainty for the estimated fidelity. Indeed, the following result shows that the minimax method gives a risk of $0.5$ (implying total uncertainty) for this measurement protocol.
\begin{proposition}
    \label{prop:minimax_method_stabilizer_insufficient_measurements}
    Let $\rho$ be an $n$-qubit stabilizer state generated by $S_1, \dotsc, S_n$. Suppose that the measurement protocol consists of measuring the first $n - 1$ generators $S_1, \dotsc, S_{n - 1}$, such that the measurements correspond to projecting on eigenspaces of $S_i$ with eigenvalue $\pm 1$ for $i = 1, \dotsc, n$. Then, irrespective of the number of shots, the minimax method gives a risk of $0.5$.
\end{proposition}
\begin{proof}
    See the end of Appendix~\ref{proof:minimax_method_stabilizer_insufficient_measurements}.
\end{proof}

We remark that if we consider eigenbasis measurements (in contrast to subspace measurements considered above), computing the risk analytically is more complicated. By eigenbasis measurement, we mean that the measurement of the stabilizer $S$ has the POVM $\{E_1, \dotsc, E_d\}$ where $E_i$ is the projection on $i^{\text{th}}$ eigenvector of $S$. Note that the eigenbasis for $S$ is not unique since $S$ has a degenerate spectrum and, therefore, the risk can depend on the choice of eigenbasis used for the measurement. For example, if the target state $\rho$ happens to coincide with one of the eigenvectors used in the measurement, then that measurement has sufficient information to accurately estimate the fidelity. However, when this does not happen, one can expect the risk to be $0.5$; see section \ref{secn:mle_pl} for an example.

\subsection{Randomized Pauli measurement scheme\label{secn:minimax_method_RPM_scheme}}
In the spirit of the randomized measurement strategy for stabilizer states we discussed above, we present a generalized version of such a Pauli measurement scheme for arbitrary (pure) states. The idea comes from the simple observation that any density matrix $\chi$ can be written as
\begin{align}
    \chi &= \frac{\id}{d} + \sum_{i = 1}^{d^2 - 1} \frac{\tr(W_i \chi)}{d} W_i \\
         &= \frac{\id}{d} + \sum_{i = 1}^{d^2 - 1} \frac{|\tr(W_i \chi)|}{d} S_i
\end{align}
where $S_i = \text{sign}(\tr(W_i \chi)) W_i$ are the Pauli operators appended with a sign. This is, in particular, valid for the target state $\rho$. Then, consider the probability distribution
\begin{equation}
    p_i = \frac{|\tr(W_i \rho)|}{\sum_{i = 1}^{d^2 - 1} |\tr(W_i \rho)|}, \quad i = 1, \dotsc, d^2 - 1 \label{eqn:randomized_pauli_measurement_probability}
\end{equation}
where $\rho$ is the target state. The measurement scheme is as follows:
\begin{tcolorbox}[colback=white, title={\hypertarget{box:pauli_measurement_scheme}{Box \ref*{secn:minimax_method}.1}: Pauli Measurement Scheme}]
\begin{enumerate}
    \item Sample a (non-identity) Pauli operator $W_i$ with probability $p_i$ ($i = 1, \dotsc, d^2 - 1$) given in Eq.~\eqref{eqn:randomized_pauli_measurement_probability} and record the outcome ($\pm 1$) of the measurement.
    \item Flip the measurement outcome $\pm 1 \to \mp 1$ if $\tr(\rho W_i) < 0$, else retain the original measurement outcome.
    \item Repeat this procedure $R$ times and feed the outcomes into the estimator given by the minimax method.
\end{enumerate}
\end{tcolorbox}
Because we exclude measurement of identity, $i$ runs from $1$ to $d^2 - 1$. We flip the outcomes because we need to measure $S_i = \text{sign}(\tr(W_i \rho)) W_i$. In Theorem~\ref{thm:minimax_method_pauli_scheme_sample_complexity}, we show how to choose the number of repetitions $R$ so as to obtain a desired value of the risk. Note that the Pauli measurement can either be a projection on the subspace with eigenvalue $\pm 1$ or a projection on the eigenvectors. For the estimator given by the minimax method, the values $+1$, $-1$ are inconsequential; all that matters is how many times $+1$ and $-1$ are observed.

Note that a very similar measurement scheme has been considered for verifying the ground state of a class of Hamiltonians \cite{takeuchi2018verification}. Further, the random sampling scheme described above is very similar to the measurement strategy used in DFE, except that we sample the Pauli operators using a different probability distribution.

For the proposed measurement strategy, the minimax method gives the following sample complexity.
\begin{theorem}
    \label{thm:minimax_method_pauli_scheme_sample_complexity}
    Let $\rho$ be an $n$-qubit pure target state. Suppose that we perform $R$ Pauli measurements as described in Box \hyperlink{box:pauli_measurement_scheme}{\ref*{secn:minimax_method}.1}. Then, for a given risk $\jnrisk \in (0, 0.5)$ and a confidence level $1 - \failure \in (0.75, 1)$,
    \begin{align}
        R &\geq 2\frac{\ln(2/\failure)}{\left|\ln\left(1 - \frac{d^2}{\mathcal{N}^2}\jnrisk^2\right)\right|} \label{eqn:minimax_method_pauli_scheme_sample_complexity} \\
          &\approx 2 \left(\frac{\mathcal{N}}{d}\right)^2 \frac{\ln(2/\failure)}{\jnrisk^2}
          \text{ when } \jnrisk^2\ll 1 \label{eqn:minimax_method_pauli_scheme_sample_complexity_approx}
    \end{align}
    measurements are sufficient to achieve the risk. Here,
    \begin{equation*}
        \mathcal{N} = \sum_{i = 1}^{d^2 - 1} |\tr(W_i \rho)|
    \end{equation*}
    and for any pure target state $\rho$, we have
    \begin{align*}
        \mathcal{N} &\leq (d - 1)\sqrt{d + 1} .
        \intertext{Therefore}
        2 \left(\frac{\mathcal{N}}{d}\right)^2 \frac{\ln(2/\failure)}{\jnrisk^2} &\leq 2 (d + 1) \left(1 - \frac{1}{d}\right)^2 \frac{\ln(2/\failure)}{\jnrisk^2}
    \end{align*}
\end{theorem}
\begin{proof}
    See the end of Appendix \ref{proof:minimax_method_pauli_scheme_sample_complexity}.
\end{proof}

Notably, the sample complexity is upper bounded by $O(d) \ln(2/\failure) / \jnrisk^2$, which is comparable to the upper bound on (the expected value of) the sample complexity in DFE (see Eq.~10 in Ref.~\cite{Flammia2011}). We can show that the sample complexity indeed improves when considering a subset of well-conditioned states defined by Flammia \& Liu \cite{Flammia2011}. To that end, suppose that for any given $i$, the state $\rho$ satisfies either $|\tr(\rho W_i)| = \alpha$ or $\tr(\rho W_i) = 0$. Then, using Eq.~\eqref{eqn:minimax_method_pauli_scheme_sample_complexity_approx} and Eq.~\eqref{eqn:pure_state_pauli_weights_constraint}, we obtain the bound of
\begin{equation}
    \frac{2}{\alpha^2} \left(\frac{d - 1}{d}\right)^2 \frac{\ln(2/\failure)}{\jnrisk^2}
\end{equation}
on the sample complexity. As before, this is comparable to DFE. In addition to giving a good sample complexity, we can obtain the estimator efficiently by reducing the optimization to a two-dimensional problem, following the ideas in Theorem \ref{thm:minimax_sample_complexity_2outcomePOVM}. Thus, if $\mathcal{N}$ can be computed efficiently, we can efficiently construct the estimator even for large dimensions.

\subsection{Robustness against errors\label{secn:minimax_method_robustness}}
For any fidelity estimator to be useful in practice, it is crucial for it to be robust to common types of experimental errors. Systematic errors in the measurements will ultimately perturb the observed outcome frequencies. Since the estimator given by the minimax method is affine, we can ensure that small changes in the observed frequencies only lead to small changes in the fidelity estimate. These changes can be quantified as follows.

Let $\bm{f}^{(1)}, \dotsc, \bm{f}^{(L)}$ denote the ``ideal" observed frequencies, i.e., when there are no measurement errors in the measurement settings $l = 1, \dotsc, L$. Let $\widetilde{\bm{f}}^{(1)}, \dotsc, \widetilde{\bm{f}}^{(L)}$ denote the actual observed frequencies due to the presence of measurement errors. From Eq.~\eqref{eqn:JN_estimator_affine}, we know that the minimax estimator can be written as
\begin{equation*}
    \widehat{F}_*(\bm{f}^{(1)}, \dotsc, \bm{f}^{(L)}) = \sum_{l = 1}^L R_l \ip{\bm{a}^{(l)}, \bm{f}^{(l)}} + c.
\end{equation*}
For simplicity, denote $\bm{f} = (\bm{f}^{(1)}, \dotsc, \bm{f}^{(l)})$ and $\bm{a} = (R_1 \bm{a}^{(1)}, \dotsc, R_L \bm{a}^{(L)})$ a larger vector obtained by concatenating all vectors from $l = 1, \dotsc, L$.
Denote $\Delta \bm{f} = \widetilde{\bm{f}} - \bm{f}$ the deviation of the observed frequencies in the erroneous measurement case from the observed frequencies in the ideal measurement case.
Using H\"older's inequality, we can infer that
\begin{align}
    |\widehat{F}_*(\bm{f}) - \widehat{F}_*(\widetilde{\bm{f}})| \leq \norm{\bm{a}}_1 \norm{\Delta \bm{f}}_\infty, \label{eqn:observed_frequency_fidelity_estimate_deviation}
\end{align}
where for any vector $\bm{x}$, we have $\norm{\bm{x}}_1 = \sum_i |x_i|$ and $\norm{\bm{x}}_\infty = \max_i |x_i|$.
We obtain
\begin{equation*}
    \norm{\bm{a}}_1 = \sum_{l = 1}^L R_l \sum_{k = 1}^{N_l} |a^{(l)}_k|
\end{equation*}
and
\begin{equation*}
    \norm{\Delta \bm{f}}_\infty = \max_{l = 1, \dotsc, L} \max_{k = 1, \dotsc, N_l} |\widetilde{f}^{(l)}_k - f^{(l)}_k|.
\end{equation*}
The quantity $\norm{\bm{a}}_1$ depends only on the minimax estimator and is finite.
Thus, for small enough perturbations, we can ensure that the fidelity estimates in the erroneous measurement case are close to the fidelity estimates in the error-free case.
This argument qualitatively shows that the estimator given by the minimax method is robust against perturbations.

We illustrate this property by considering a simple example of an erroneous measurement setup that one could encounter in practice.
Suppose we wish to perform Pauli measurements, but due to some issue with the experimental setup,
each qubit undergoes a systematic rotation at the beginning of the measurement process, changing the POVM.
The ideal measurements correspond to measuring Pauli operators and recording the observed eigenvalues.
This procedure amounts to choosing POVMs that project onto the $\pm 1$ eigenspace of the respective Pauli operator (we call this subspace projection).
To simulate erroneous measurements, we suppose that each qubit undergoes an erroneous rotation along a given axis by a fixed angle.
For the purposes of our example, we generate a randomly chosen axis (unit vector) $\bm{n}_j$ for the $j$th qubit.
The rotation angle $\theta_j$ for the $j$th qubit is chosen from a normal distribution with zero mean and standard deviation $\sigma_\theta$ that controls the amount of noise.
Thus, the unitary describing this erroneous rotation is given by $U_{\text{error}} = \otimes_{j = 1}^n \exp(i \theta_j \bm{n}_j \cdot \bm{\sigma})$, where $\bm{\sigma} = (X, Y, Z)$ denotes the $1$-qubit Pauli operators.
We assume that this error is systematic, i.e., the same error acts on all the POVM elements.
In other words, if the ideal POVMs are $E^{(l)}_k$, the actual POVMs being implemented are $U_{\text{error}} E^{(l)}_k U_{\text{error}}^\dagger$.
Since we choose $E^{(l)}_{\pm} = (\id \pm P_l) / 2$ as projections onto eigenspaces of the Pauli operator $P_l$, the erroneous POVM corresponds to $(\id + U_{\text{error}} P_l U_{\text{error}}^\dagger) / 2$.
This is just a slight rotation of the Pauli observable to be measured,  due to calibration issues for example.

To study the performance of our fidelity estimation method under such a systematic error, we consider a $3$-qubit GHZ state as the target state, and the actual state is obtained by applying $10 \%$ depolarizing noise to the target state. We measure $L = 7$ (non-identity) Pauli operators that have non-zero expectation with the target state. Each of these operators is measured $R_l = 300$ times. We choose a confidence level of $95 \%$. In Table~\ref{tab:JN_robustness}, we list the fidelity estimate given by our method using ideal Pauli measurements, the fidelity estimate due to erroneous Pauli measurements, as well as the bound on the deviation of the estimates given by Eq.~\eqref{eqn:observed_frequency_fidelity_estimate_deviation}. We find that the deviations in the fidelity estimate always lie within this bound. In particular, for small rotation angles, the fidelity estimates with and without error are nearly the same.
This substantiates our claim that our method is robust to small amounts of measurement errors without any modifications to the estimation procedure.

\begin{table}[!ht]
    \begin{center}
        \begin{tabular}{c >{\centering\arraybackslash}p{2.3cm} c c c}
            \toprule
            $\sigma_\theta$ & estimates from erroneous measurements & deviation & bound \\[2pt]
            \midrule
            0.05 rad (\hphantom{0}$2.86^\circ$) & 0.911 & 0.004 & 0.004 \\
            0.1\hphantom{0} rad (\hphantom{0}$5.73^\circ$) & 0.901 & 0.014 & 0.058 \\
            0.15 rad (\hphantom{0}$8.59^\circ$) & 0.878 & 0.037 & 0.069 \\
            0.2\hphantom{0} rad ($11.45^\circ$) & 0.849 & 0.066 & 0.111 \\
            \bottomrule
        \end{tabular}
    \end{center}
    \caption{Analysis of fidelity estimates given by minimax method under systematic measurement error. The true fidelity with a $3$-qubit GHZ target state is $0.912$. Ideal Pauli measurements give a fidelity estimate of $0.915$. A systematic measurement error is present due to erroneous rotation of Pauli operators, with the size of the rotation controlled by $\sigma_\theta$. The calculations in the table correspond to one randomly drawn instance of these erroneous rotations. Fidelity estimates using erroneous measurements, the absolute value of deviation of estimates from erroneous and ideal measurements, and the bound given by Eq.~\eqref{eqn:observed_frequency_fidelity_estimate_deviation} are listed as a function of $\sigma_\theta$.}
    \label{tab:JN_robustness}
\end{table}

Depending on the experimental setup, the measurement noise encountered in the system can be different from the one we have considered above. However, even for other forms of noise, our method is guaranteed to be robust, although the exact dependence on the noise may vary. If one is interested in precise quantitative bounds for a specific error model, these can be computationally estimated following the procedure outlined above.

\section{Comparison with other methods\label{secn:other_methods_comparison}}
We now compare the minimax method with two commonly used techniques to estimate the fidelity: direct fidelity estimation (DFE) and maximum likelihood estimation (MLE) (and a related approach called Profile Likelihood (PL)). We further compare it with a simple semidefinite programming (SDP) based approach.

\subsection{Direct fidelity estimation}
Flammia \& Liu \cite{Flammia2011} and Silva \textit{et al.}\ \cite{DaSilva2011} construct a fidelity estimator by judiciously sampling Pauli measurements. If the target state $\rho$ is well-conditioned, i.e., $|\tr(\rho W)| \geq \alpha$ or $\tr(\rho W) = 0$ for each Pauli operator $W$ and some fixed $\alpha > 0$, then DFE gives a good sample complexity. Specifically, their method uses $O(\log(1/\failure)/\alpha^2 \adderr^2)$ measurement outcomes to obtain an estimate for fidelity within an additive error of $\adderr$ and a confidence level of $1 - \failure$ \cite{Flammia2011}. Their rigid measurement scheme, however, can have disadvantages in practice, compared to a more flexible approach that works for arbitrary POVMs.

Consider a random $4$-qubit state as the target state, with the actual state obtained by applying $10\%$ depolarizing noise to the target state. We choose a risk of $0.05$ in DFE and a confidence level of $95\%$, and obtain the Pauli measurements that need to be performed as prescribed by DFE. However, to perform these Pauli measurements, we choose two different types of POVM: $(1)$ projection on $+1$ and $-1$ eigenvalue subspaces of the Pauli operator, and $(2)$ projection on each eigenstate of the Pauli operator, which is common in experiments. For DFE, it doesn't matter much which of the two POVMs is implemented because the estimator there only uses the expectation value, while the minimax method makes a distinction between these POVMs. From Table~\ref{tab:DFE_JN_comparison}, we can see that the risk ($\approx 0.023$) for the minimax method is less than half the DFE risk ($0.05$) when using projection on subspaces $(1)$. When projecting on each eigenbasis element $(2)$, we obtain a risk ($\approx 0.012$), which is even lower than before, and clearly much better than the DFE risk. That is, we are able to use the larger expressive power of POVM $(2)$ compared to POVM $(1)$ to lower the risk. Since the minimax method has already computed an estimator with a low risk (as per the prescription of the DFE method), all subsequent experiments can use the \textit{same} measurement settings to estimate the fidelity. Random sampling of Pauli measurements is not necessary.

\begin{table}[!ht]
    \begin{center}
        \begin{tabular}{l c c c c}
            \toprule
            \multicolumn{5}{c}{Random state} \\[2pt]
            \cmidrule{2-5}
             & \multicolumn{2}{c}{Subspace projection} & \multicolumn{2}{c}{Eigenbasis projection} \\[2pt]
             \cmidrule(r){2-3}
             \cmidrule(l){4-5}
             & DFE & Minimax & DFE & Minimax \\[2pt]
            \midrule
            True fidelity & \multicolumn{2}{c}{0.906} & \multicolumn{2}{c}{0.906} \\
            Estimate & 0.895 & 0.895 & 0.900 & 0.902 \\
            Risk & 0.050 & 0.023 & 0.050 & 0.012 \\[2pt]
            \cmidrule[\heavyrulewidth]{1-5}
            \multicolumn{5}{c}{GHZ state} \\[2pt]
            \cmidrule{2-5}
             & \multicolumn{2}{c}{Subspace projection} & \multicolumn{2}{c}{Eigenbasis projection} \\[2pt]
             \cmidrule(r){2-3}
             \cmidrule(l){4-5}
             & DFE & Minimax & DFE & Minimax \\[2pt]
            \midrule
            True fidelity & \multicolumn{2}{c}{0.906} & \multicolumn{2}{c}{0.906} \\
            Estimate & 0.907 & 0.907 & 0.904 & 0.906 \\
            Risk & 0.050 & 0.022 & 0.050 & 0.018 \\[2pt]
            \bottomrule
        \end{tabular}
    \end{center}
    \caption{Comparison of DFE method with the minimax method for a $4$-qubit random state and $4$-qubit GHZ state as target states. The Pauli measurements are performed as per the prescription of DFE, but two different types of POVMs are used: projection on subspace with $+1$ and $-1$ eigenvalue, and projection on each eigenstate of the Pauli operator. We find that the minimax method has lower risk in most cases, while the estimates are comparable to those of DFE method.}
    \label{tab:DFE_JN_comparison}
\end{table}

Next, we consider a $4$-qubit GHZ state as the target state, which is a well-conditioned state as per DFE. As before, we perform measurements as prescribed by DFE by choosing a risk of $0.05$ for DFE method and a confidence level of $95\%$. We summarize the estimated fidelity and risks in the second half of Table~\ref{tab:DFE_JN_comparison}. Similarly to the case of the random state, we observe that the risk given by the minimax method for the GHZ state is always lower than the DFE risk. Changing from POVM $(1)$ to POVM $(2)$ again leads to an improvement in the risk.

Finally, we argue that the Pauli measurement scheme prescribed by DFE can be far from optimal for certain target states. To demonstrate this, we choose a rather extreme example of a random $4$-qubit target state with the following measurement scheme: perform as many measurements and repetitions as prescribed by DFE for achieving a risk of $0.05$ at a confidence level of $95\%$, but instead of performing Pauli measurements, we choose random POVMs, generated as per Ref.~\cite{heinosaari2020random}. In this case, we obtain a risk of $\approx 0.023$ from the minimax method, which is larger than the risk for Pauli measurements using the minimax method (see Tab.~\ref{tab:DFE_JN_comparison}), but still a factor of 2 lower than the DFE risk. This suggests that when the target state is random, performing random measurements is almost as good as performing Pauli measurements. Hence there is no real advantage in using DFE in terms of sample complexity. If, instead, we used the optimal (though impractical) measurement scheme described in section \ref{secn:minimax_method_optimal_risk}, we can obtain a risk of $0.05$ using just $\approx 700$ total measurements (independent of the target state and dimension) as opposed to $\approx 248000$ total measurements required by DFE for the random target state. While this random state example is rather atypical, it stands to demonstrate that there are cases where the DFE measurement scheme is outperformed by a more tailored scheme, which the minimax method can benefit from.

A potential disadvantage for the minimax method is that once the dimension of the system becomes very large, depending on the measurement scheme, constructing the estimator can become inefficient. One reason for this inefficiency is that the algorithm we use for optimization requires projection on the set of density matrices. Because we perform this projection through diagonalization, the projection can cost up to $O(d^3)$ iterations, where $d$ is the dimension of the system. In contrast, DFE can handle large dimensions well because the classical computations involved are simple and most of the complexity is absorbed into performing the experiments. This drawback of the minimax method is alleviated if we use the Pauli measurement scheme described in Box \hyperlink{box:pauli_measurement_scheme}{\ref*{secn:minimax_method}.1}, which is similar to DFE. We have written an efficient algorithm to construct the estimator given by the minimax method for this measurement scheme, as noted in section~\ref{secn:minimax_method_RPM_scheme}. Moreover, the estimator can be pre-computed once the measurement scheme has been defined.

\subsection{MLE and Profile Likelihood\label{secn:mle_pl}}
Maximum Likelihood Estimation (MLE) is a popular approach used for quantum tomography \cite{Hradil1997}. One can think of MLE as minimizing the Kullback-Liebler divergence between the observed frequencies and the Born probabilities \cite{vrehavcek2001iterative}. This naturally leads to the (negative) log-likelihood function
\begin{equation}
    \ell\left(\{\bm{f}^{(l)}\}_{l = 1}^L \big| \chi\right) = - \sum_{l = 1}^L \frac{R_l}{R} \sum_{k = 1}^{N_l} f^{(l)}_k \ln(\tr(E^{(l)}_k \chi)), \label{eqn:MLE_negative_log_likelihood}
\end{equation}
where $R = \sum_{l = 1}^L R_l$.
One then minimizes the function $\ell$ (or equivalently, maximizes the likelihood function) to obtain the MLE estimate $\widehat{\sigma}$ for the experimental quantum state. We calculate $\tr(\rho \widehat{\sigma})$ to estimate the fidelity with a target state $\rho$. While the estimates are usually good enough when a sufficient number of measurements is used for the reconstruction, a major disadvantage is that the method provides no confidence intervals for the estimated fidelity.

A common approach to compute an uncertainty of the MLE estimate is Monte-Carlo (MC) re-sampling. Herein, one numerically generates (artificial) outcomes based on the observed frequencies and the expected statistical noise distribution, referred to as re-sampling. For each set of outcomes one reconstructs the state and estimates the fidelity as before. This process is repeated many times to obtain a large number of MLE fidelity estimates which can be used to calculate an (asymmetric) interval around the median that corresponds to the chosen confidence level. If necessary, hedging \cite{Blume-Kohout2010HedgedMLE} can be implemented to deal with zero probabilities. Note that the MC approach is similar to the nonparametric bootstrap method used for finding confidence intervals \cite{diciccio1996bootstrap}.

A disadvantage of such an MLE based approach is that we need to reconstruct the state, which can be costly. Moreover, there are certain freedoms in the re-sampling and hedging definitions, which can affect the results. We argue that the MC approach can give overconfident uncertainty bounds. This issue can be seen in the following example: take a $2$-qubit Bell state, stabilized by $XX$ and $ZZ$, as the target state. We measure only $XX$ with $500$ repetitions, which is not enough to uniquely determine the state. The true state is obtained by applying $10\%$ depolarizing noise to the Bell state, so that the actual fidelity is $0.925$. We generate $100$ MLE fidelity estimates (with new measurement outcomes generated at every repetition), and we find that the average MLE fidelity estimate is $0.44$ when the measurements correspond to projection onto an eigenbasis of $XX$.

The large error (an estimate of $0.44$ compared to the true value of $0.925$) is expected since the measurements are chosen poorly with respect to the state. This problem, however, is not detected by the MC approach, where we find that the average uncertainty corresponding to a confidence level of $95\%$ is $(0.32, 0.38)$. This indicates that MC uncertainty is overconfident (because $0.44 + 0.38 = 0.82 < 0.925$). Indeed, this is seen through the empirical coverage probability which turns out to be $0$, in stark contrast with the high confidence level of $95\%$ that was chosen. By empirical coverage probability, we mean the fraction of cases where the true fidelity lies inside the computed confidence interval. We obtain very similar results when using subspace measurements (average fidelity of $0.45$, average uncertainty of $(0.31, 0.37)$, and empirical coverage probability of $0$). Therefore, we cannot always trust the uncertainty given by the MC approach. In contrast, the minimax approach gives a risk of $0.5$ for this example (implying total uncertainty) when insufficient number of stabilizer measurements are provided irrespective of whether eigenbasis or subspace measurements are used.

We note that such a problem with the MC approach is also present in more realistic scenarios. For example, consider a $3$-qubit W-state as the target state, and suppose that the true state has a fidelity of $99.1\%$ with the target state. We measure $N = 19$ Pauli operators (subspace measurements), choosing all the non-identity Pauli operators that have non-zero overlap with the $W$-state. Each Pauli measurement is repeated $100$ times, and the outcomes together are used to estimate the fidelity corresponding to a confidence level of $95 \%$. Using a total of $150$ such estimates, we find that the MC method gives an empirical coverage probability of $72 \%$, which is much smaller than the chosen confidence level of $95 \%$. Thus, the MC approach is overconfident and gives tighter bounds than warranted in practical situations of interest.

In addition to the statistical problems noted above, we remark that MC re-sampling is also computationally costly. This is because MC re-sampling requires performing MLE for each (artificial) fidelity estimate generated, which is costly. Moreover, this costly computation needs to be performed after the experimental data has been collected. In contrast, our method only needs to compute the estimator once (potentially costly computation), and this can be done before collecting measurement data. Once the estimator has been constructed, it runs very efficiently on the experimental data. Thus, we not only obtain statistical but also computational benefits over MLE with MC re-sampling.

An alternate approach to obtain a bound for the MLE estimate is calculating the Profile Likelihood (PL) function \cite{murphy2000profile}. Given any value $F \in [0, 1]$, PL corresponds to the solution of the optimization problem
\begin{align}
    \text{PL}(F) = &\min_{\chi \in \mathcal{X}}\ \ell\left(\{\bm{f}^{(l)}\}_{l = 1}^L \big| \chi\right) \nonumber \\
                   &\text{s.t.} \quad \tr(\rho \chi) = F.
\end{align}
Note that we define the profile likelihood in terms of negative log-likelihood instead of the likelihood function. Our definition of PL amounts to solving the MLE optimization problem, except for the added constraint that the fidelity with the target state must be equal to $F$. The MLE solution can be obtained from PL by adding an additional layer of optimization: $\text{MLE} = \min_{F \in [0, 1]} \text{PL}(F)$. It can be shown that $\text{PL}(F)$ is convex in $F$. The advantage of calculating PL is that given a cut-off value, one can obtain a bound on the fidelity similar to an error bar. To this end, we draw a horizontal line (the cut-off) on the PL vs $F$ plot, and the locations along the $F$ axis where this line intersects the curve gives a bound on estimated fidelity. The MLE estimate lies inside this interval because it corresponds to the minimum. A schematic of a typical PL plot is shown in Fig.~\ref{fig:PL_schematic}. A somewhat similar idea for obtaining confidence regions was proposed by Faist \& Renner~\cite{faist2016practical}.
\begin{figure}[ht]
    \begin{center}
        \begin{tikzpicture}
          \draw[->] (-0.5, 0) -- (3.5, 0) node[below = 5pt] {$F$};
          \draw[->] (0, -0.5) -- (0, 3) node[left = 5pt] {PL};
          \draw[scale=0.5, domain=0.75:4, smooth, variable=\x, black] plot ({\x}, {(\x - 2)*(\x - 3.5) + 1.5});
          \draw[dashed, red] (-0.5, 1) -- (3, 1) node {};
          \draw[dashed, blue] (0.86, 1) -- (0.86, 0) node[below = 7pt, left = 0.5pt] {$F_1$};
          \draw[dashed, blue] (1.9, 1) -- (1.9, 0) node[below = 7pt, right = 0.5pt] {$F_2$};
          \draw[dashed, dkgreen] (1.4, 0.47) -- (1.4, 0) node[below = 2pt] {{\scriptsize MLE}};
        \end{tikzpicture}
    \end{center}
    \caption{A schematic of a typical Profile Likelihood (PL) plotted against the parameter $F$ (representing fidelity). The MLE estimate (green) corresponds to the minimum of the PL curve. A cut-off for PL (red dashed horizontal line) gives a bound $[F_1, F_2]$ for the estimated fidelity.}
    \label{fig:PL_schematic}
\end{figure}
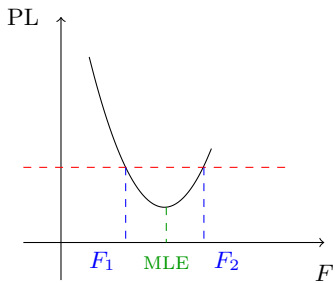

Using PL is a natural way of providing a bound for MLE, since it returns an interval of fidelity estimates that correspond to large enough likelihood. However, it does not provide a true confidence interval as the location of the cut-off value is unknown.
The issue is that likelihoods can be considered a generic notion of ``plausibility'' but not, in the technical sense, ``uncertainty;'' and that 
PL gives a \emph{likelihood region}, which is not the same as a coverage probability~\cite{Sprott2000}.

What we want is that given many MLE estimates (from actual experimental data), the true fidelity must lie $1 - \failure$ fraction of times inside the bound calculated from PL, where $1 - \failure$ is the chosen confidence level. One general heuristic that is used to give such a cut-off for PL is based on Wilks' Theorem~\cite{wilks1938large}. However, this theorem is only valid asymptotically as the sample size becomes arbitrarily large, and furthermore, it doesn't work well for quantum states, especially those that are low rank~\cite{Scholten2016, glancy2012gradient}. Note that Scholten \& Blume-Kohout~\cite{Scholten2016} give an alternative to Wilks' Theorem for quantum states, but like Wilks' Theorem, their alternative is only exact asymptotically. Furthermore, it uses models different from what we need for PL (they consider nested Hilbert spaces of increasing dimension, while we require density matoices of the same dimension but with an added constraint). Therefore, as of now, we do not know how to obtain a confidence interval using PL.

Hence, for the purpose of demonstration, we choose different cut-off values for the PL and compute the bounds (as described below) on the MLE estimate corresponding to each chosen value. We take a $3$-qubit random target state and apply $10\%$ depolarizing noise to obtain the actual state. We measure $L = 48 = 0.75 \times 4^3$ Pauli operators, chosen in the decreasing order of weights given by DFE, repeating each measurement $100$ times. Using the knowledge of the true state, we check in which of these bounds the true fidelity lies. This process is repeated $10^3$ times by generating different observed frequencies using the true state, and subsequently, we obtain the coverage probability corresponding to each value of cut-off. We also find the average width of the bound for each cut-off value. We plot the coverage probability against this average width in Fig.~\ref{fig:PL_JN_SDP_ROC}, finding that the bounds are reasonably tight. Note that the true fidelity is usually not known, so the coverage probability cannot be computed in practice.

\subsection{An SDP-based approach}
Finally, we compare with a simple Semi-Definite Programming (SDP) method to obtain a bound on fidelity. This involves solving the following intuitive optimization problem to obtain bounds on fidelity.
\begin{align}
    &F_{\min} = \min_{\chi \in \mathcal{X}} \tr(\rho \chi) \nonumber \\
    &\hspace{1cm} \text{s.t.}\ \sum_{i = 1}^L \sum_{k = 1}^{N_i} \left(\tr(E^{(i)}_k \chi) - f^{(i)}_k\right)^2 \leq \epsSDP \\
    &F_{\max} = \max_{\chi \in \mathcal{X}} \tr(\rho \chi) \nonumber \\
    &\hspace{1cm} \text{s.t.}\ \sum_{i = 1}^L \sum_{k = 1}^{N_i} \left(\tr(E^{(i)}_k \chi) - f^{(i)}_k\right)^2 \leq \epsSDP
\end{align}
In essence, we find the minimum and maximum fidelity with the target state over density matrices that satisfy the measurement statistics up to an error of $\epsSDP$. The advantage of such an approach is that the bounds obtained are independent of the method used to estimate fidelity. The drawback, however, is that the parameter $\epsSDP$ needs to be chosen ``by hand". Since we need to tune the parameter $\epsSDP$ similar to the cut-off in PL method, we plot the coverage probability against average width of the bound as before in the Fig.~\ref{fig:PL_JN_SDP_ROC}. We use the same state and measurement settings as for PL, and $10^3$ estimates to compute the coverage probability. We can see that the average width of the bound given by SDP method is typically larger than both PL and the minimax method.

\begin{figure}[!ht]
    \includegraphics[width=0.5\textwidth]{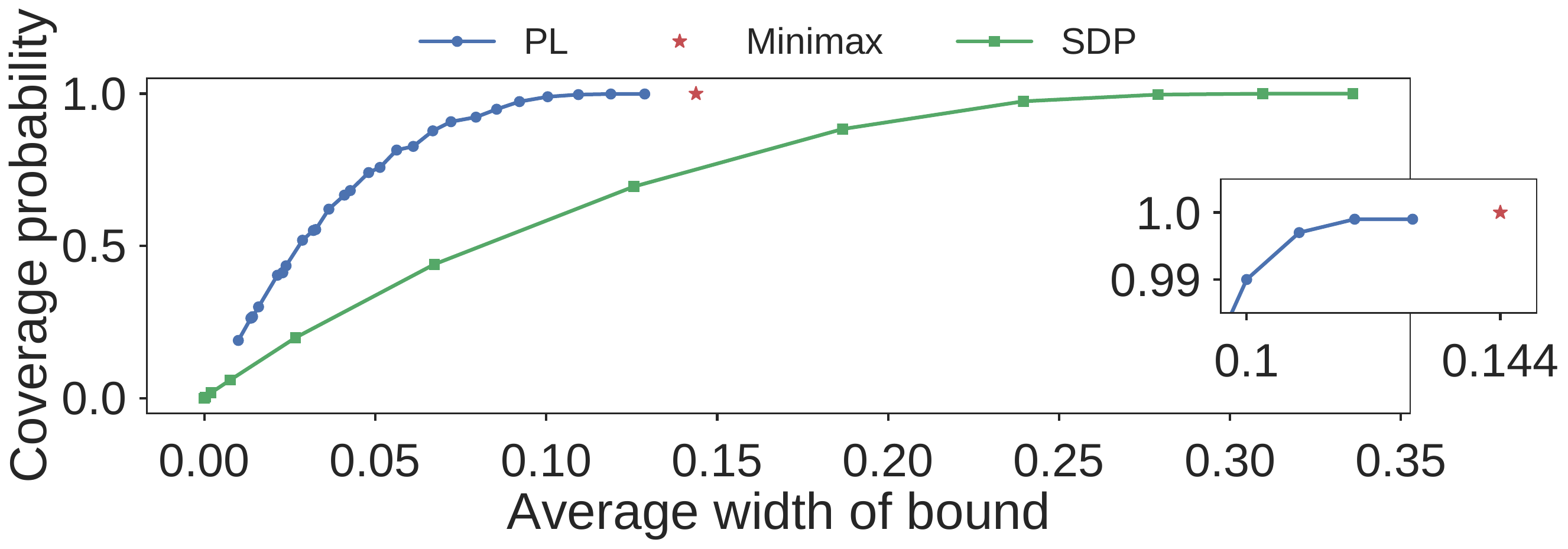}
\caption{Plot of empirical coverage probability against the average width of the bound given by the Profile Likelihood (PL) and Semi-Definite Programming (SDP) methods for a $3$-qubit random target state and $48$ Pauli measurements. It can be seen that PL gives tighter bounds than the SDP method. For reference, the minimax bound corresponding to a confidence level of $95\%$ is shown as well.}
    \label{fig:PL_JN_SDP_ROC}
\end{figure}

We remark that the minimax method gives a wider confidence interval than PL in Fig.~\ref{fig:PL_JN_SDP_ROC} due to its conservative definition. In turn, the minimax method's confidence interval is guaranteed to hold irrespective of what the true state of the system is. Indeed, we can construct the fidelity estimator using the minimax method even before taking any data. In contrast, PL and SDP methods compute bounds on the fidelity after the experimental data is collected, so they can be tighter in principle. On the other hand, the PL method requires choosing a cut-off while the SDP method requires choosing the parameter $\epsSDP$. For experimental data, we have no systematic way to choose these quantities to ensure that the bound corresponds to a genuine confidence interval. Therefore, we cannot use PL and SDP methods to generate confidence intervals in practice. In contrast, the minimax method gives rigorous confidence intervals without free parameters.

Table \ref{tab:fidelity_estimaiton_methods_comparison} provides a quick overview of comparison of these different methods.
{
\renewcommand{\arraystretch}{1.5}
\begin{table*}[ht]
    \begin{tabular}{>{\raggedright}p{4cm}p{2cm}p{2cm}p{1.8cm}p{1.5cm}p{1.8cm}}
    & Minimax & MLE, MC & PL & SDP & DFE \tabularnewline
    \hline
        No unknown parameters required? &  \yes$^a$ & \yes & \no & \no & \yes \tabularnewline
        Provides a rigorous confidence interval (never overconfident)? & \yes & \no & \no$^b$ & \no$^b$ & \yes \tabularnewline
        Risk level known before seeing the outcomes? & \yes & \no & \no & \no & \yes \tabularnewline
        Applies to any measurement setting? & \yes & \yes & \yes & \yes & \no  \tabularnewline
        No significant computation required in practice? & \no$^c$ & \no & \no & \no & \yes \tabularnewline
    \hline
    \end{tabular}
    
    \begin{minipage}[t]{0.75\textwidth}
        \begin{flushleft}
            $^a${\footnotesize \parbox[t]{\textwidth}{\raggedright Technically, $\epsZero$ can be considered as a free parameter, but we fix it at $\epsZero = 10^{-5}$. Since we don't have to tune $\epsZero$, we don't list it as an unknown parameter.}}
            $^b${\footnotesize Because no systematic method of obtaining a confidence interval is known.}\\
            $^c${\footnotesize \parbox[t]{\textwidth}{\raggedright The computational time required for the minimax method depends on the system dimension, the target state, the measurement settings, and the algorithm used. For the Pauli measurement scheme in section~\ref{secn:minimax_method_RPM_scheme}, we have a fast algorithm irrespective of the target state or the dimension.}}
        \end{flushleft}
    \end{minipage}

    \caption{Comparison of different methods for estimation of fidelity: Minimax method, Maximum Likelihood Estimation (MLE) with Monte-Carlo (MC) sampling for estimating the fidelity and uncertainty, Profile Likelihood (PL) and Semi-Definite Programming (SDP) methods for calculating bounds on fidelity, and the direct fidelity estimation (DFE) method.}
    \label{tab:fidelity_estimaiton_methods_comparison}
\end{table*}
}

\subsection{Quantum state verification\label{secn:qsv_comparison}}
Quantum state verification (QSV) is a procedure to certify that the fidelity of the quantum state prepared in the lab with a pure target state is greater than some pre-specified value with high probability~\cite{pallister2018optimal, zhu2019efficient, wang2019optimal, jiang2020towards}. QSV corresponds to a hypothesis testing problem, which is related to, but different from, our problem of estimation of fidelity. Nevertheless, the approach adopted by Pallister \textit{et al.}~\cite{pallister2018optimal} and related/subsequent work~\cite{zhu2019efficient, zhu2019general, wang2019optimal, jiang2020towards} shares similarities with our fidelity estimation method, and therefore, a comparison of these methods is helpful.

Suppose that $\rho = \op{\psi}{\psi}$ is the pure target state one wishes to prepare, but one prepares the states $\sigma_1, \sigma_2, \dotsc, \sigma_R$ in the lab. For comparison with our method, we assume that $\sigma_i = \sigma$ are identically prepared states and that measurements are independent.
One is given the promise that either $F(\rho, \sigma) = 1$  or $F(\rho, \sigma) \leq 1 - \qsverr$ for some $\qsverr \in (0, 1)$.
If one can reject the hypothesis $F(\rho, \sigma) \leq 1 - \qsverr$ with probability $1 - \failure$, then we can infer that $F(\rho, \sigma) = 1$ with a confidence level of $1 - \failure$.
The error $\qsverr$ denotes (a bound on) the deviation of the true fidelity from the maximum fidelity of $1$.
This is different from the additive error $\jnrisk$ in fidelity estimation, which 
bounds the deviation of the true fidelity from the estimated fidelity.

The goal of QSV~\cite{pallister2018optimal} is to obtain a minimax optimal protocol for rejecting the hypothesis $F(\rho, \sigma) \leq 1 - \qsverr$. Pallister \textit{et al.}~\cite{pallister2018optimal} consider random measurement strategies to achieve this goal, and show that one can obtain the minimax optimal strategy by solving an optimization problem.
They demonstrate their method by obtaining a minimax optimal measurement strategy for arbitrary two-qubit states. In particular, to reject $F(\rho, \sigma) \leq 1 - \qsverr$ with a confidence level of $1 - \delta$, their strategy needs
\begin{equation*}
    O\left(\frac{\ln(1/\failure)}{\qsverr}\right)
\end{equation*}
samples. The scaling of $1/\qsverr$ is optimal for QSV under the aforementioned assumptions.
We note that quantum state certification~\cite{buadescu2019quantum} also achieves $O(1/\qsverr)$ scaling but using joint measurements, which are difficult to implement with current technology. In contrast, Pallister \textit{et al.}~\cite{pallister2018optimal} present a local, single-copy measurement protocol for achieving this scaling for two-qubit states.

We show in Appendix~\ref{app:qsv_comparison} that using the same measurement strategy, our method can estimate fidelity to an additive error $\jnrisk$ with a confidence level of $1 - \delta$ using
\begin{equation*}
    O\left(\frac{\ln(2/\failure)}{\jnrisk^2}\right)
\end{equation*}
samples. Note that the scaling of $1/\jnrisk^2$ is optimal for any fidelity estimation protocol using only independent single-copy measurements~\cite{Huang2020}.
Thus, we demonstrate for the two-qubit example, that the optimal QSV measurement protocol is also optimal for fidelity estimation using our method.  

Notice, however, that the scaling of the optimal sample complexity is different for QSV ($O(1/\qsverr)$) and fidelity estimation ($O(1/\jnrisk^2)$).
Such a difference boils down to the differences in problem setup for QSV and fidelity estimation (for example, the different meanings of $\qsverr$ and $\jnrisk$)
and the underlying assumptions.
For instance, if we remove the assumption of QSV that either $F(\rho, \sigma) = 1$  or $F(\rho, \sigma) \leq 1 - \qsverr$, as one would in an actual experiment, the scaling of QSV can be worse than $1/\epsilon_q$ (see Ref.~\cite{jiang2020towards} and its supplementary material).
Despite these differences, the approach adopted in QSV can help with designing optimal measurement protocols that our method can benefit from for fidelity estimation, as demonstrated by the two-qubit example above.
Additional details on QSV and how it compares with our method can be found in Appendix~\ref{app:qsv_comparison}.

\subsection{Classical shadows}
Recently, a method called classical shadows was introduced for the estimation of expectation values of observables~\cite{Huang2020}.
We begin by remarking that the philosophy of classical shadows differs from our method.
In classical shadows, one seeks to find an optimal estimation procedure, assuming that there is no a priori information about the observables whose expectation values need to be computed.
In contrast, we seek to find the optimal estimation procedure, given a priori knowledge of the observable, which is often known in practice.
In the context of fidelity estimation, we know the target state $\rho$ with which we wish to estimate the fidelity.

To obtain a classical shadow, one randomly selects a unitary $U$ from a given ensemble $\mathcal{U}$ according to some probability distribution,
rotates the state by $U$ (i.e., $\sigma \to U \sigma U^\dagger$), and performs a computational basis measurement. 
If $\ket{b}$ is the computation basis element observed after one such measurement, then in expectation, we obtain a quantum channel $\mathcal{M}_c$, defined by $\mathbb{E}[U^\dagger \op{b}{b} U] = \mathcal{M}_c(\sigma)$, where $\sigma$ is the underlying state.
If the unitary ensemble is tomographically complete, then the channel $\mathcal{M}_c$ is invertible~\cite{Huang2020},
and we call $\mathcal{M}_c^{-1}(U^\dagger \op{b}{b} U)$ a classical snapshot of $\sigma$.
After $R$ measurements, the array of classical snapshots $\left\{\mathcal{M}_c^{-1}(U_1^\dagger \op{b_1}{b_1} U_1), \,\ldots\,, \mathcal{M}_c^{-1}(U_R^\dagger \op{b_r}{b_r} U_R)\right\}$ (where $U_i$ is the randomly sampled unitary and $\ket{b_i}$ is the computational basis element observed at the $i$th time step) is called a classical shadow.
Then, given an observable $\mathcal{O}$, an estimate of the expectation value $\tr(\mathcal{O} \sigma)$ using a single classical snapshot is given by $\tr(\mathcal{O} \mathcal{M}_c^{-1}(U_1^\dagger \op{b_1}{b_1} U_1))$. Multiple estimates are combined using the median of means to give the final estimate~\cite{Huang2020}.

Classical shadows are minimax optimal up to a constant factor in the worst-case over all observables, and thus, are suitable to use when the observable is not known a priori.
Our method is minimax optimal up to a constant factor for a given observable and measurement protocol.
Thus, given the observable whose expectation needs to be estimated and the measurement protocol being implemented, our method can match the performance (sample complexity) of classical shadows up to a constant factor.

Both these approaches have their respective advantages.
Since classical shadows use a fixed randomized measurement protocol, they can simultaneously estimate $N$ observables to a precision of $\adderr$ with $O(\mathcal{S} \log(N)/ \adderr^2)$ samples, where $\mathcal{S}$ is the maximum shadow norm of the observables~\cite{Huang2020}.
However, since the observables are not assumed to be known beforehand, the measurements are not tailored to the observables being estimated.
In contrast, for a given observable, we are allowed to choose an appropriate measurement protocol to estimate the expectation value of that observable.
As a result, our method, in principle, can use exponentially fewer samples than classical shadows for estimating the expectation value.
A simple example is estimating the expectation with a Pauli observable acting non-trivially on all qubits.
The
pre-factor (shadow norm) $\mathcal{S}$ scales exponentially using randomized global Clifford measurements as well as randomized Pauli measurements, as was pointed out in Huang \textit{et al.}~\cite{Huang2020}. As a result, one has to take exponentially many samples to estimate this expectation value using classical shadows.
In contrast, measuring in the eigenbasis of the given Pauli observables only needs $O(1/\adderr^2)$ samples to estimate the expectation to an error of $\adderr$.

For fidelity estimation with a pure target state $\rho$, we begin by studying the case when the unitary ensemble $\mathcal{U}$ is the Clifford group $\mathcal{C}_n$ on $n$ qubits, and we sample uniformly at random from $\mathcal{C}_n$. It has been shown~\cite{Huang2020} that classical shadows only need $O(1/\adderr^2)$ measurements to estimate the fidelity with $\rho$ to an error of $\epsilon$, which matches the optimal sample complexity up to a constant factor.
Owing to minimax optimality, our method can match the sample complexity of classical shadows using global Clifford measurements using $O(1/\adderr^2)$ measurements.
However, since the size of the Clifford group scales as $|\mathcal{C}_n| = 2^{n^2 + 2n} \prod_{i = 1}^n (4^i - 1)$ with the number of qubits $n$~\cite{koenig2014efficiently}, it is computationally challenging to compute our estimate for global Clifford measurements by directly optimizing over POVM elements.
To circumvent this problem, we need to simplify the optimization problem involved in computing our estimator analytically, as done, for example, for the randomized Pauli measurement scheme described in Sec.~\ref{secn:minimax_method_RPM_scheme}. We leave the problem of simplifying the optimization to a future study.

We remark that the computational difficulty with computing estimates is not specific to our method, but also shared by classical shadows. For Clifford measurements, the classical shadow stores the stabilizer states $U_i^\dagger \op{b_k}{b_k} U_i$, which can be done efficiently. However, the target state (or more generally, the observable $\mathcal{O}$) with which the expectation needs to be computed need not be a stabilizer state (or Clifford unitary). For efficient computation of the estimate $\tr(\mathcal{O} U_i^\dagger \op{b_k}{b_k} U_i)$, one needs to be able to write the observable as a linear combination of polynomially many Cliffords (or stabilizer states). If this is not possible, there may be no efficient way of storing the observable in memory or computing the estimate.
Thus, for many observables and target states of interest, even classical shadows suffers from exponential classical computation time, albeit in a different manner than our method.
We remark that for stabilizer states, we have an efficient algorithm to compute our estimator that also achieves optimal sample complexity as described in Sec.~\ref{secn:minimax_method_stabilizer_states}.
We conjecture this algorithm can be generalized to states that can expressed as a linear combinations of polynomially many stabilizer states, with at most a polynomial overhead.

Note, however, that Clifford measurements can be considered as an exceptional case, especially in the near term. Not only does the size of the Clifford group grow extremely quickly with qubit number, even the experimental complexity for realizing n-qubit Clifford gates grows rapidly, making this a rather impractical approach.
In practice, experiments almost always rely on local measurements, such as Pauli measurements.

For fidelity estimation using Pauli measurements, we can get an exponential advantage over classical shadows.
A simple example is estimating the fidelity with the target state $\ket{0...0}$ (computational basis state with all zeros).
Since the shadow norm using Pauli measurements scales exponentially with the locality of the observable~\cite{Huang2020}, and because $\ket{0...0}$ has a non-trivial support on all $n$ qubits, classical shadows need exponentially many samples to estimate the fidelity with this state using random Pauli measurements.
In contrast, our method (and even DFE~\cite{Flammia2011}) can estimate the fidelity with $\ket{0...0}$, using computational basis measurements that does not scale with the number of qubits.
This is because the measurement protocol is tailored to the target state, which is known a priori.

\section{Concluding remarks and future research}
The minimax method can be used to obtain an estimator for the fidelity with a pure target state for any measurement scheme. For a given setting, the estimator only needs to be computed once, and can subsequently be evaluated on raw measurement outcomes instantaneously. This gives our method a practical advantage over other estimation protocols which require random sampling of measurement settings.

Crucially, the minimax method not only constructs an estimator but also provides rigorous confidence intervals that are nearly minimax optimal. We showed that this property translates to practical sample complexity when the measurement scheme is carefully chosen. Notably, the risk is a property of the chosen measurement scheme (including the number of repetitions) and target state and is thus computed before seeing any data. Because the risk is pre-computed, it is taken to be symmetric around the fidelity estimate. As a consequence, it might be sub-optimal in practice, as the confidence interval might include unphysical values when the estimate is close to 0 or 1.

On the other hand, the fact that the risk is known beforehand allows us to use the method for benchmarking experimental protocols without having to take any data. This can, therefore, be useful in guiding the design of experiments. Further, when extending the method to quantum channels using the Choi-Jamio\l{}kowski isomorphism~\cite{choi1975completely, jamiolkowski1972linear, jiang2013channel}, such benchmarking can also be done for protocols that estimate gate fidelity.

The computation involved in finding the fidelity estimator is practical for relatively small dimensions, but can become computationally intensive for larger systems because intermediate steps in our algorithm need $O(d^3)$ iterations. Therefore, finding more efficient ways to do the optimization needed to find the estimator would prove very useful in practice. We show that this is possible for a specific measurement protocol involving Pauli measurements (see section \ref{secn:minimax_method_RPM_scheme}), where the optimization is reduced to two-dimensions irrespective of dimension of the system. The only challenge is efficiently computing the Pauli weights for the measurement protocol, but that is not a drawback of the algorithm itself. For example, when these weights can be efficiently computed, as possible for well-conditioned states, the estimator can be efficiently computed in very large dimensions. It would be interesting to extend such an approach to more general measurement settings in order to efficiently compute the fidelity estimator. Similarly, approaches to compute the estimator efficiently for randomized Clifford measurements needs to be investigated. This problem essentially reduces to efficiently computing the classical fidelity between two probability distributions determined by randomized Clifford measurements on two different states.
Furthermore, it would be interesting to see if using conjugate gradient methods for the initial descent and then using accelerated projected gradient methods like Nesterov's method~\cite{nesterov1988approach} can give significant improvements in performing the optimization for target states and measurement settings. Such an approach was proposed by Ref.~\cite{shang2017superfast} to give a fast algorithm for MLE.

An obvious advantage of the minimax method over tomographic methods like MLE is that the state need not be reconstructed. We showed that the uncertainty given by Monte-Carlo method (for MLE), although typically tighter, need not correspond to a genuine confidence interval as it can be overconfident. PL also gives tighter bounds than the minimax method because it is computed after seeing the data, but we are not aware of any systematic method to obtain a confidence interval using PL. Finding an approach to obtain confidence intervals from PL and similarly for the SDP approach outlined above are interesting directions for future research. Finally the confidence intervals from the minimax method, though guaranteed to be correct, are often not as tight as they could be. This could potentially be improved by generalizing to an asymmetric risk or by computing the risk after seeing the data, or both.

\section{Code availability}
An open source implementation of the minimax method can be found on \href{https://github.com/akshayseshadri/minimax-fidelity-estimation}{https://github.com/akshayseshadri/minimax-fidelity-estimation}.

\begin{acknowledgments}
The authors thank Emanuel Knill, Scott Glancy, Yanbao Zhang, and Rainer Blatt for helpful discussions on the manuscript.

This material is based upon work supported by the National Science Foundation under Grant Nos. 1819251 and 2112901.
This work utilized the Summit supercomputer, which is supported by the National Science Foundation (awards ACI-1532235 and ACI-1532236), the University of Colorado Boulder, and Colorado State University. The Summit supercomputer is a joint effort of the University of Colorado Boulder and Colorado State University.

We gratefully acknowledge support by the Austrian Science Fund (FWF Grant-DOI 10.55776/F71) (SFB BeyondC) and the Institut f\"ur Quanteninformation GmbH. We also acknowledge funding from the EU H2020-FETFLAG-2018-03 under Grant Agreement no. 820495, by the Office of the Director of National Intelligence (ODNI), Intelligence Advanced Research Projects Activity (IARPA), via US Army Research Office (ARO) grant no. W911NF-16-1-0070 and W911NF-21-1-0007, and the US Air Force Office of Scientific Research (AFOSR) via IOE Grant No. FA9550-19-1-7044 LASCEM. This project has received funding from the European Union’s Horizon 2020 research and innovation programme under the Marie Skłodowska-Curie grant agreement No 840450.
\end{acknowledgments}

\bibliographystyle{apsrev}
\bibliography{}

\appendix
\section{Minimax Method: Theory\label{app:minimax_theory}}
We discuss here the theory surrounding the minimax method for fidelity estimation. We begin by giving a short overview of Juditsky \& Nemirovski's framework \cite{Juditsky2009} in appendix~\ref{app:JN_premise}. Then, in appendix~\ref{app:JN_fidelity_estimation}, we describe how we adapt their general method for the purpose of fidelity estimation so as to obtain the procedure described in section~\ref{secn:minimax_method_theory}.

\subsection{Juditsky \& Nemirovski's premise\label{app:JN_premise}}
Suppose we are given a set of ``states" $\mathcal{X} \subseteq \mathbb{R}^{d}$ that is a compact and convex subset of $\mathbb{R}^d$. We wish to estimate the linear functional $\ip{g, x}$, where $g \in \mathbb{R}^d$ is some fixed vector, while the state $x \in \mathcal{X}$ of the system is unknown to us. We do not have direct access to the state $x$. However, we have access to a single measurement outcome determined by the state $x$. Measurements are modelled using random variables that assign probabilities to the possible outcomes depending on the state.

To that end, Juditsky \& Nemirovski~\cite{Juditsky2009} consider a family of random variables $\mathrm{Z}_\mu$ parametrized by $\mu \in \mathcal{M}$, where $\mathcal{M} \subseteq \mathbb{R}^m$ is some set of parameters. These random variables take values in a Polish space\footnote{A Polish space is a topological space that is homeomorphic to a separable complete metric space. We endow this with the Borel $\sigma$-algebra.} $(\Omega, \Sigma)$ equipped with a $\sigma$-finite Borel measure $\mathbb{P}$~\cite{Juditsky2009}. We assume that $\mathrm{Z}_\mu$ has a probability density $p_\mu$ with respect to this reference measure $\mathbb{P}$. The state $x \in \mathcal{X}$ determines the random variable $\mathrm{Z}_{A(x)}$ through an affine function $A\colon \mathcal{X} \to \mathcal{M}$, and we are given one outcome of this random variable for the purpose of estimation. Looking ahead to the specialization of this framework to the quantum setting, one can think of the set of observations $\Omega$ as a finite set, Borel measurable functions on $\Omega$ as functions on $\Omega$, and the integrals $\int_\Omega f(\omega) d\mathbb{P}$ as the finite sum $\sum_{\omega \in \Omega} f(\omega)$.

Our goal is to construct an estimator for $\ip{g, x}$ that uses an outcome of the random variable $\mathrm{Z}_{A(x)}$ to give an estimate. An estimator is a real-valued Borel measurable function on $\Omega$. The set of estimators $\mathcal{F}$ we are allowed to work with is any finite-dimensional vector space comprised of real-valued Borel measurable functions on $\Omega$ as long as it contains constant functions~\cite{Juditsky2009}. The mapping $\mathcal{D}(\mu) = p_\mu$ between the parameter $\mu$ and the corresponding probability density $p_\mu$ is called a parametric density family~\cite{Juditsky2009}. In order to be able to choose an appropriate estimator from $\mathcal{F}$ given that the probability density of the random variable is $p_{A(x)}$, we want that the set of estimators $\mathcal{F}$ and the parametric density family $\mathcal{D}$ interact well with each other. This gives rise to the notion of a good pair defined by Juditsky \& Nemirovski~\cite{Juditsky2009}.
\begin{definition}[Good pair]
    \label{defn:good_pair}
    We call a given pair $(\mathcal{D}, \mathcal{F})$ of parametric density family $\mathcal{D}$ and finite-dimensional space $\mathcal{F}$ of Borel functions on $\Omega$ a good pair if the following conditions hold.
    \begin{enumerate}
        \item $\mathcal{M}$ is a relatively open convex set in $\mathbb{R}^m$. By relatively open, we mean $\mathcal{M} = \text{relint}(\mathcal{M}) \equiv \{x \in \mathcal{M} \mid \exists r > 0 \text{ with } B(x, r) \cap \text{aff}(\mathcal{M}) \subseteq \mathcal{M}\}$, where $\text{aff}(\mathcal{M})$ is the affine hull of $\mathcal{M}$.
        \item Whenever $\mu \in \mathcal{M}$, we have $p_\mu(\omega) > 0$ for all $\omega \in \Omega$.
        \item Whenever $\mu, \nu \in \mathcal{M}$, $\phi(\omega) = \ln(p_\mu(\omega)/p_\nu(\omega)) \in \mathcal{F}$.
        \item Whenever $\phi \in \mathcal{F}$, the function
                \begin{equation*}
                    F_{\phi}(\mu) = \ln\left(\int_\Omega \exp\left(\phi(\omega)\right) p_\mu(\omega) d\mathbb{P}\right)
                \end{equation*}
              is well-defined and concave in $\mu \in \mathcal{M}$.
    \end{enumerate}
\end{definition}
Any estimator in $\widehat{g} \in \mathcal{F}$ is called an affine estimator (note, however, that $\widehat{g}$ need not be an affine function).

To judge the performance of an arbitrary estimator $\widehat{g}$, we define the $\failure$-risk as follows~\cite{Juditsky2009}.
\begin{definition}[$\failure$-risk]
    Given a confidence level $1 - \failure \in (0, 1)$, we define the $\failure$-risk associated with an estimator $\widehat{g}$ as
    \begin{equation*}
        \mathcal{R}(\widehat{g}; \failure) = \inf\left\{\error: \sup_{x \in \mathcal{X}} \Prob_{\omega \sim p_{A(x)}}\left\{\omega: |\widehat{g}(\omega) - \ip{g, x}| > \error\right\} < \failure\right\},
    \end{equation*}
    where $\omega \sim p_{A(x)}$ means that $\omega$ is sampled according to $p_{A(x)}$. The corresponding minimax optimal risk is defined as
    \begin{equation*}
        \mathcal{R}_*(\failure) = \inf_{\widehat{g}} \mathcal{R}(\widehat{g}; \failure)
    \end{equation*}
    where the infimum is taken over \textit{all} Borel functions $\widehat{g}$ on $\Omega$. Restricting to just the affine estimators, the affine risk is defined as
    \begin{equation*}
        \mathcal{R}_{\text{aff}}(\failure) = \inf_{\widehat{g} \in \mathcal{F}} \mathcal{R}(\widehat{g}; \failure) .
    \end{equation*}
\end{definition}
It turns out that we don't lose much by restricting our attention to affine estimators. Indeed, Juditsky \& Nemirovski~\cite{Juditsky2009} prove that if $(\mathcal{D}, \mathcal{F})$ is a good pair, there is an estimator $\widehat{F}_* \in \mathcal{F}$ with $\failure$-risk at most $\jnrisk$, such that
\begin{align*}
    \mathcal{R}_{\text{aff}}(\failure) &\leq \jnrisk \leq \vartheta(\failure) \mathcal{R}_*(\failure) \\
    \vartheta(\failure) &= 2 + \frac{\ln(64)}{\ln(0.25/\failure)}
\end{align*}
for $\failure \in (0, 0.25)$.

The estimator $\widehat{F}_*$ and the risk $\jnrisk$ are constructed as follows~\cite{seshadri2021computation, Juditsky2009}.
\begin{enumerate}[leftmargin=0.2cm]
    \item Find the saddle-point value of the function $\Phi\colon (\mathcal{X} \times \mathcal{X}) \times (\mathcal{F} \times \mathbb{R}_+) \to \mathbb{R}$ defined as
    \begin{align}
        \Phi(x, y;&\ \phi, \alpha) = \ip{g, x} - \ip{g, y} + 2\alpha \ln(2/\failure) \nonumber \\
                                    &+ \alpha \Bigg[\ln\left(\int_\Omega \exp(-\phi(\omega)/\alpha) p_{A(x)}(\omega) d\mathbb{P}\right) \nonumber \\
                                    &\hspace{1cm} + \ln\left(\int_\Omega \exp(\phi(\omega)/\alpha) p_{A(y)}(\omega) d\mathbb{P}\right)\Bigg] . \label{eqn:Phi_general}
    \end{align}
    to a given precision. Juditsky \& Nemirovski~\cite{Juditsky2009} show that $\Phi$ has the following properties. $\Phi$ is concave in $(x, y)$ and convex in $(\phi, \alpha)$, and also $\Phi \geq 0$. Further, $\Phi$ has a well-defined saddle-point. See Ref.~\cite{Juditsky2009} for other properties and a more general treatment of the problem.

    \item Denote the saddle-point value of $\Phi$ by $2\jnrisk$:
        \begin{align}
            \jnrisk &= \frac{1}{2} \sup_{x, y \in \mathcal{X}} \inf_{\phi \in \mathcal{F}, \alpha > 0} \Phi(x, y; \phi, \alpha) \nonumber \\
                                    &= \frac{1}{2} \inf_{\phi \in \mathcal{F}, \alpha > 0} \max_{x, y \in \mathcal{X}} \Phi(x, y; \phi, \alpha) \label{eqn:JN_risk_saddle_point_general}.
        \end{align}
        Say the saddle-point value is achieved at some points $x^*, y^* \in \mathcal{X}$, $\phi_* \in \mathcal{F}$ and $\alpha_* > 0$ to a precision $\error~>~0$. Suppose that an outcome $\omega \in \Omega$ is observed upon measurement of the random variable $\mathrm{Z}_{A(x)}$. Then, the estimator $\widehat{F}_* \in \mathcal{F}$ is given as
        \begin{align}
            &\widehat{F}_*(\omega) = \phi_*(\omega) + c \nonumber
            \intertext{where the affine estimator $\phi_*$ is given by}
            &\frac{\phi_*}{\alpha_*} = \frac{1}{2} \ln\left(\frac{p_{A(x^*)}}{p_{A(y^*)}}\right) \label{eqn:phi_alpha_opt}
            \intertext{and the constant $c$ is}
            &c = \frac{1}{2} \left(\ip{g, x^*} + \ip{g, y^*}\right) \label{eqn:JN_estimator_constant_general}
        \end{align}
        The $\failure$-risk associated with this estimator satisfies $\mathcal{R}(\widehat{F}_*; \failure) \leq \jnrisk + \error$, so the final output of the procedure is $\widehat{F}_*(\omega) \pm (\jnrisk + \error)$. See Refs.~\cite{seshadri2021computation, Juditsky2009} for details.
\end{enumerate}
Importantly, the estimator $\widehat{F}_*$ is a function that can accept any outcome $\omega \in \Omega$. In other words, the estimate is provided depending on the outcome, but the risk $\jnrisk$ is computed before seeing any outcome.

So far we have described the one-shot scenario, i.e., producing an estimate for $\ip{g, x}$ from one outcome of a single random variable $\mathrm{Z}_{A(x)}$. In practice, we will need to consider outcomes of different random variables $\mathrm{Z}_{A^{(l)}(x)}$, which corresponds to $l = 1, \dotsc, L$ different measurement settings. More precisely, we are given Polish spaces $(\Omega^{(l)}, \Sigma^{(l)})$ equipped with a $\sigma$-finite Borel measure $\mathbb{P}^{(l)}$ for $l = 1, \dotsc, L$. We are also given a set of parameters $\mathcal{M}^{(l)}$ for $l = 1, \dotsc, L$. For each $l = 1, \dotsc, L$, we are given a family of random variables $\mathrm{Z}_{\mu_l}$ taking values in $\Omega^{(l)}$, where $\mu_l \in \mathcal{M}^{(l)}$. The random variable $\mathrm{Z}_{\mu_l}$ has a probability density $p^{(l)}_{\mu_l}$ with respect to the reference measure $\mathbb{P}^{(l)}$. As before, we are given affine mappings $A^{(l)}\colon \mathcal{X} \to \mathcal{M}^{(l)}$ for $l = 1, \dotsc, L$ that map the state $x \in \mathcal{X}$ of the system to a corresponding parameter. For each $l = 1, \dotsc, L$, we can choose estimators for the $l^{\text{th}}$ measurement from the set $\mathcal{F}^{(l)}$, which is a finite-dimensional vector space of real-valued Borel measurable functions on $\Omega^{(l)}$ that contains constant functions. To incorporate the outcomes of these different random variables, Juditsky \& Nemirovski~\cite{Juditsky2009} define the direct product of good pairs, which essentially constructs one large good pair from many smaller ones.
\begin{definition}[Direct product of good pairs]
    \label{defn:direct_product_good_pair}
    Considering the following quantities for $l = 1, \dotsc, L$. Let $(\Omega^{(l)}, \Sigma^{(l)})$ be a Polish space endowed with a Borel $\sigma$-finite measure $\mathbb{P}^{(l)}$. Let $\mathcal{D}^{(l)}(\mu_l) = p^{(l)}_{\mu_l}$ be the parametric density family for $\mu_l \in \mathcal{M}^{(l)}$. Let $\mathcal{F}^{(l)}$ be a finite-dimensional linear space of Borel functions on $\Omega^{(l)}$ containing constants, such that the pair $(\mathcal{D}^{(l)}, \mathcal{F}^{(l)})$ is good. Then the direct product of these good pairs $(\mathcal{D}, \mathcal{F}) = \bigotimes_{l = 1}^L (\mathcal{D}^{(l)}, \mathcal{F}^{(l)})$ is defined as follows.
    \begin{enumerate}
        \item The large space is $\Omega = \Omega^{(1)} \times \dotsb \times \Omega^{(L)}$ endowed with the product measure $\mathbb{P} = \mathbb{P}^{(1)} \times \dotsb \times \mathbb{P}^{(L)}$.
        \item The set of parameters is $\mathcal{M} = \mathcal{M}^{(1)} \times \dotsb \times \mathcal{M}^{(L)}$, and the associated parametric density family is $\mathcal{D}(\mu) = p_\mu \equiv \prod_{l = 1}^L p^{(l)}_{\mu_l}$ for $\mu =(\mu_1, \dotsc, \mu_L) \in \mathcal{M}$.
        \item The linear space $\mathcal{F}$ comprises of all functions $\phi$ defined as $\phi(\omega_1, \omega_2, \dotsc, \omega_L) = \sum_{l = 1}^L \phi^{(l)}(\omega_l)$, where $\phi^{(l)} \in \mathcal{F}^{(l)}$ and $\omega_l \in \Omega^{(l)}$ for $l = 1, \dotsc, L$.
    \end{enumerate}
\end{definition}
It can be verified that the direct product of good pairs is a good pair~\cite{Juditsky2009}. Therefore, we can apply the above procedure for constructing an optimal estimator to the direct product of good pairs to obtain an optimal estimator that accounts for all the given measurement outcomes.

We now note a simplification of the risk $\jnrisk$ obtained using results of Juditsky \& Nemirovski \cite{Juditsky2009}. Recall that we defined the risk $\jnrisk$ as half the saddle-point value of the function $\Phi$ (see Eq.~\eqref{eqn:JN_risk_saddle_point_general}). However, this definition is not very amenable to theoretical calculations. Therefore, we note an alternate expression for the risk, given by Proposition~3.1 of Juditsky \& Nemirovski~\cite{Juditsky2009}.
\begin{equation}
    \jnrisk = \frac{1}{2} \max_{x, y \in \mathcal{X}} \left\{\ip{g, x} - \ip{g, y} \mid \affh(A(x), A(y)) \geq \frac{\failure}{2}\right\} \label{eqn:JNriskprop3.1_general}
\end{equation}
Juditsky \& Nemirovski~\cite{Juditsky2009} refer to the quantity
\begin{equation}
    \affh(\mu, \nu) = \int_\Omega \sqrt{p_\mu p_\nu} d\mathbb{P} \label{eqn:Hellinger_affinity}
\end{equation}
as Hellinger Affinity (this quantity is sometimes referred to as Bhattacharyya coefficient in the discrete case; see for example Ref.~\cite{fuchs1999cryptographic}). Juditsky \& Nemirovski~\cite{Juditsky2009} prove that the Hellinger affinity is continuous and log-concave in $(\mu, \nu) \in \mathcal{M} \times \mathcal{M}$ when $(\mathcal{D}, \mathcal{F})$ is a good pair.

\subsection{Fidelity estimation\label{app:JN_fidelity_estimation}}
Now we give the details on how we adapt Juditsky \& Nemirovski's general framework to fidelity estimation. In Table~\ref{tab:JN_quantities_quantum}, we provide a dictionary relating the general quantities defined in appendix~\ref{app:JN_premise} to our scenario of fidelity estimation. We assume that the quantum system has a $d$-dimensional Hilbert space over $\mathbb{C}$, $d \in \mathbb{N}$. The set of $d \times d$ complex-valued Hermitian matrices forms a real vector space that is isomorphic to $\mathbb{R}^{d^2}$, and thus, we can write $\mathcal{X} \subseteq \mathbb{R}^{d^2}$, where $\mathcal{X}$ is the set of density matrices. Note that $\mathcal{X}$ is a compact and convex set. For the $l^\text{th}$ measurement setting, we consider the positive operator-valued measure (POVM) $\{E^{(l)}_1, \dotsc, E^{(l)}_{N_l}\}$, where $l = 1, \dotsc, L$. Since the definition of a good pair requires that the probability of each outcome is non-zero, we add a small parameter $0 < \epsZero \ll 1$ to make the outcome probabilities positive, as noted in section~\ref{secn:minimax_method_theory}. These outcome probabilities are represented by the affine map $A^{(l)}\colon \mathcal{X} \to \mathcal{M}^{(l)}$ given in Table~\ref{tab:JN_quantities_quantum}.

\begin{table}[!ht]
    \begin{center}
        \begin{tabular}{l l}
        \toprule
            $\mathcal{X}$ & Set of density matrices \\
            $\Omega^{(l)}$ & Measurement outcomes $\{1, \dotsc, N_l\}$ \\
            $\mathbb{P}^{(l)}$ & Counting measure on $(\Omega^{(l)}, \Sigma^{(l)})$ with $\Sigma^{(l)} = 2^{\Omega^{(l)}}$ \\
            $\mathcal{M}^{(l)}$ & Relatively open simplex $\{x \in \mathbb{R}^{N_l} \mid x_i > 0,\ \sum_i x_i = 1\}$ \\
            $p_\mu$ & $p_\mu = (\mu_1, \dotsc, \mu_{N_l})$, $\mu \in \mathcal{M}^{(l)}$ \\
            $A^{(l)}$ & $A^{(l)}(\chi)_k = \frac{\tr(E^{(l)}_k \chi) + \epsZero/N_l}{1 + \epsZero}$, $k = 1, \dotsc, N_l$, $\epsZero > 0$ \\
            $\mathcal{F}^{(l)}$ & Set of estimators: real-valued functions on $\Omega^{(l)}$ \\
            $g$ & Pure target state $\rho$ \\
            \bottomrule
        \end{tabular}
    \end{center}
    \caption{A dictionary specifying the meaning of each quantity appearing in the text for the purpose of fidelity estimation. The index $l$ varies from $1$ to $L$, where $L$ denotes the number of measurement settings. The $l^\text{th}$ measurement setting is described by the POVM $\{E^{(l)}_1, \dotsc, E^{(l)}_{N_l}\}$.}
    \label{tab:JN_quantities_quantum}
\end{table}

Note that the set of outcomes $\Omega^{(l)}$ is a finite set for each $l = 1, \dotsc, L$. We consider the discrete topology $\Sigma^{(l)} = 2^{\Omega^{(l)}}$ on $\Omega^{(l)}$, so that $(\Omega^{(l)}, \Sigma^{(l)})$ forms a Polish space. The Borel $\sigma$-algebra coincides with the topology, so we use the same symbol for both. Because $\Sigma^{(l)}$ is discrete, any real-valued function on $\Omega^{(l)}$ is Borel measurable, and therefore, we omit the phrase ``Borel measurable" when talking about functions (or estimators) on $\Omega^{(l)}$. Furthermore, since each $\Omega^{(l)}$ is a finite set, real-valued functions defined on it can be considered as $|\Omega^{(l)}| = N_l$ dimensional real vectors. Thus, we treat elements of $\mathcal{F}^{(l)}$ as vectors.

Using these facts, we show that $\mathcal{D}^{(l)}$ and $\mathcal{F}^{(l)}$ as defined in Table~\ref{tab:JN_quantities_quantum} form a good pair. By definition of $\mathcal{M}^{(l)}$ and $p_\mu$ in Table~\ref{tab:JN_quantities_quantum}, it is easy to see the first and second conditions for good pair given in definition~\ref{defn:good_pair} hold. We check that the last two conditions given in definition~\ref{defn:good_pair} hold. Since $\mathcal{F}^{(l)}$ contains all functions on $\Omega^{(l)}$, it contains $\ln(p_\mu/p_\nu)$ in particular. Next, we see that for $\phi^{(l)} \in \mathcal{F}^{(l)}$, we can write
\begin{equation*}
    F_{\phi^{(l)}}(\mu) = \ln\left(\sum_{k = 1}^{N_l} \exp\left(\phi^{(l)}_k\right) \mu_k\right)
\end{equation*}
because $\mathbb{P}^{(l)}$ is the counting measure. As noted above, we consider $\phi^{(l)} \in \mathcal{F}^{(l)}$ as an $N_l$-dimensional real vector. Then, since $\nabla_\mu^2 F_{\phi^{(l)}} = -\bm{e} \bm{e}^T/(\ip{\bm{e}, \mu})^2$, where $\bm{e} = (e^{\phi^{(l)}_1}, \dotsc, e^{\phi^{(l)}_{N_l}})$, $F_{\phi^{(l)}}$ is concave in $\mu$. Thus, $(\mathcal{D}^{(l)}, \mathcal{F}^{(l)})$ forms a good pair.

Now, suppose that we perform $R_l$ repetitions (shots) of the $l^\text{th}$ measurement setting. Then, the space to be considered for all measurement settings put together is $\Omega = (\Omega^{(1)})^{R_1} \times \dotsm \times (\Omega^{(L)})^{R_L}$, the parameter space for probability distributions is $\mathcal{M} = (\mathcal{M}^{(1)})^{R_1} \times \dotsm \times (\mathcal{M}^{(L)})^{R_L}$, and the set of estimators $\mathcal{F} \subseteq (\mathcal{F}^{(1)})^{R_1} \times \dotsm \times (\mathcal{F}^{(L)})^{R_l}$ is chosen as per definition~\ref{defn:direct_product_good_pair}. The mapping $A: \mathcal{X} \to \mathcal{M}$ is given by $A(\chi) = \bigoplus_{l = 1}^L \bigoplus_{r = 1}^{R_l} A^{(l)}(\chi) \equiv (A^{(1)}(\chi), A^{(1)}(\chi), \dotsc, A^{(L)}(\chi), A^{(L)}(\chi))$, where $A^{(l)}(\chi)$ is repeated $R_l$ times. Then, we can use the direct product of good pairs (definition~\ref{defn:direct_product_good_pair}) to compute the function $\Phi$ defined in Eq.~\eqref{eqn:Phi_general} when all the measurement settings are considered together.

Note that $\phi \in \mathcal{F} \subseteq (\mathcal{F}^{(1)})^{R_1} \times \dotsm \times (\mathcal{F}^{(L)})^{R_L}$ implies $\phi = \sum_{l = 1}^L \sum_{r = 1}^{R_l} \phi^{(l, r)}$, where $\phi^{(l, r)}$ belongs to the $r^\text{th}$ copy of $\mathcal{F}^{(l)}$. Then, using Eq.~\eqref{eqn:Phi_general}, we obtain
\begin{align*}
    \Phi(&\chi_1, \chi_2; \phi, \alpha) = \tr(\rho \chi_1) - \tr(\rho \chi_2) + 2\alpha \ln(2/\failure) \nonumber \\
                        &+ \alpha \sum_{l = 1}^L \sum_{r = 1}^{R_l} \Bigg[\ln\left(\sum_{k = 1}^{N_l} e^{-\phi^{(l, r)}_k/\alpha} \frac{\tr(E^{(l)}_k \chi_1) + \epsZero/N_l}{1 + \epsZero}\right) \nonumber \\
                                            &\hspace{1cm} + \ln\left(\sum_{k = 1}^{N_l} e^{\phi^{(l, r)}_k/\alpha} \frac{\tr(E^{(l)}_k \chi_2) + \epsZero/N_l}{1 + \epsZero}\right)\Bigg]
\end{align*}
where we have used the fact that $\exp(\sum_{l, r} \phi^{(l, r)}) = \prod_{l, r} \exp(\phi^{(l, r)})$, $p_\mu = \prod_{l, r} p^{(l)}_{\mu_{l, r}}$, $A(\chi) = \bigoplus_{l, r} A^{(l)}(\chi)$, and that $\mathbb{P} = (\mathbb{P}^{(1)})^{R_1} \times \dotsm \times (\mathbb{P}^{(L)})^{R_L}$ is a product measure. By remark $3.2$ of Juditsky \& Nemirovski \cite{Juditsky2009}, we can just use $R_l$ copies of $\phi^{(l)}$ instead of $\phi^{(l, r)}$ in finding the saddle-point of the above function, and thus we obtain Eq.~\eqref{eqn:Phi_quantum}.

Suppose that the $(\chi_1, \chi_2)$ and the $\alpha$ components of the saddle-point of the function $\Phi$ are attained at $\chi_1^*, \chi_2^* \in \mathcal{X}$ and $\alpha_* > 0$, respectively, to a given precision. Then, from Eq.~\eqref{eqn:phi_alpha_opt}, definition~\ref{defn:direct_product_good_pair}, and preceding remarks, we can infer that the $\phi$-component of the saddle-point can be described by the function $\phi_* = \sum_{l = 1}^L \sum_{r = 1}^{R_l} \phi^{(l)}_*$, where for $l = 1, \dotsc, L$, we have
\begin{equation}
    \phi^{(l)}_* = \frac{\alpha_*}{2} \ln\left(\frac{p_{A^{(l)}(\chi_1^*)}}{p_{A^{(l)}(\chi_2^*)}}\right). \label{eqn:phi_opt_quantum}
\end{equation}

Note that replacing $g = \rho$ with $g = \mathcal{O}$ for any Hermitian operator (observable) $\mathcal{O}$, we can obtain an estimator for the expectation value of that observable.

Finally, to adapt the simplified expression for risk given in Eq.~\eqref{eqn:JNriskprop3.1_general} to the quantum case, we compute the Hellinger affinity (see Eq.~\eqref{eqn:Hellinger_affinity}). The Hellinger affinity for the quantum case is given as
\begin{align}
    \affh(A(\chi_1), A(\chi_2)) &= \prod_{l = 1}^L \Bigg[\sum_{k = 1}^{N_l} \left(\frac{\tr(E^{(l)}_k \chi_1) + \epsZero/N_l}{1 + \epsZero}\right)^{1/2} \nonumber \\
                                &\hspace{1.5cm} \left(\frac{\tr(E^{(l)}_k \chi_2) + \epsZero/N_l}{1 + \epsZero}\right)^{1/2}\Bigg]^{R_l} \label{eqn:Hellinger_affinity_quantum} \\
                                &\approx \prod_{l = 1}^L \left[F_C(\chi_1, \chi_2; \{E^{(l)}_k\})\right]^{R_l/2} \nonumber
\end{align}
where in the last step, we neglect $\epsZero \ll 1$ to simplify the equations (we include $\epsZero$ in the numerical simulations). Here, $F_C$ is the classical fidelity defined in Eq.~\eqref{eqn:classicalfidelity}. Then substituting the expression for the Hellinger affinity in Eq.~\eqref{eqn:JNriskprop3.1_general}, we obtain Eq.~\eqref{eqn:JNriskprop3.1}.

\section{Minimax Method: Numerical Implementation\label{app:minimax_numerical}}
We outline the procedure followed to find the saddle-point of the function $\Phi$ defined in Eq.~\eqref{eqn:Phi_quantum}, from which we can compute the fidelity estimator.

We present two algorithms to compute the saddle point. The first algorithm is based on Nesterov's method (accelerated projected gradient), and the second algorithm uses standard \texttt{cvxpy} software. The implementation based on Nesterov's method is more memory efficient, whereas the implementation based on \texttt{cvxpy} is faster in practice. A detailed comparison of computation time required by each method can be found in accompanying Ref.~\cite{PRL}. We now describe both our algorithms and show that they converge to the saddle point.

\subsection{Implementation using Nesterov's method}
This is done in two steps: find the $\chi_1^*, \chi_2^* \in \mathcal{X}$ and $\alpha_* > 0$ components of the saddle-point and then compute $\phi_*/\alpha_*$ to obtain the $\phi_* \in \mathcal{F}$ component of the saddle-point.

To execute the first step, we resort to the following expression for the saddle-point value noted in Ref.~\cite{seshadri2021computation}, which is obtained by appropriately re-writing the expression given by Juditsky \& Nemirovski~\cite{Juditsky2009},
\begin{align}
    2\jnrisk &= \inf_{\alpha > 0} \bigg\{2 \alpha \ln(2/\failure) + \max_{\chi_1, \chi_2 \in \mathcal{X}} \bigg[\ip{\rho, \chi_1} - \ip{\rho, \chi_2} \nonumber \\
                             &\hspace{3.5cm}+ 2 \alpha \ln(\affh(A(\chi_1), A(\chi_2)))\bigg]\bigg\} \label{eqn:saddle_point_expression_optimization}
\end{align}
where the Hellinger affinity $\affh(A(\chi_1), A(\chi_2))$ is given in Eq.~\eqref{eqn:Hellinger_affinity_quantum}.

Now, we note that as per the definitions used in Tab.~\ref{tab:JN_quantities_quantum}, the function $\affh(\mu, \nu) = \prod_{l = 1}^L \left[\sum_{i = 1}^{N_l} \sqrt{\mu^{(l)}_i \nu^{(l)}_i}\right]^{R_l}$ is smooth on its domain because $\mu, \nu > 0$ (component-wise) for each $\mu, \nu \in \mathcal{M}$. Since $\affh(\mu, \nu) > 0$, the function $\ln(\affh(\mu, \nu))$ is well-defined and smooth on its domain. Further, since $A\colon \mathcal{X} \to \mathcal{M}$ is affine, $\ln(\affh(A(\chi_1), A(\chi_2)))$ is smooth on $\mathcal{X} \times \mathcal{X}$. In particular, the derivatives of $\ln(\affh(A(\chi_1), A(\chi_2)))$ are continuous. Since $\mathcal{X} \times \mathcal{X}$ is compact, the Hessian of $\ln(\affh(A(\chi_1), A(\chi_2)))$ is bounded, and therefore, the gradient of $\ln(\affh(A(\chi_1), A(\chi_2)))$ is Lipschitz continuous. Moreover, $\ln(\affh(A(\chi_1), A(\chi_2)))$ is a jointly concave function of the density matrices (see Appendix~\ref{app:JN_premise}).

With this in mind, we use the following procedure to find a saddle-point of the function $\Phi$ defined in Eq.~\eqref{eqn:Phi_quantum} to any given precision.
\begin{enumerate}[leftmargin=0.2cm]
    \item For any fixed $\alpha > 0$, we solve the ``inner" convex optimization problem in Eq.~\eqref{eqn:saddle_point_expression_optimization}
            \begin{align*}
                \max_{\chi_1, \chi_2 \in \mathcal{X}} &\bigg[
                        \underbrace{\tr(\rho \chi_1) - \tr(\rho \chi_2) + 2 \alpha \ln(\affh(A(\chi_1), A(\chi_2)))}_{f(\chi_1, \chi_2)}
                    \bigg] \\
                    = \max_{\chi \in \overline{\mathcal{X}}}\ &f(\chi) \nonumber
            \end{align*}
            using the version of Nesterov's second method \cite{nesterov1988approach} given in Ref.\ \cite{tseng2010approximation}. For the second equation above, we define $\chi = (\chi_1, \chi_2)$ and $\overline{\mathcal{X}} = \mathcal{X} \times \mathcal{X}$. Nesterov's second method is suited to problems where the objective $f$ is a convex function\footnote{Or a concave function in the case of maximization.} with a Lipschitz continuous gradient (see Theorem 1 (c) in Ref.~\cite{tseng2010approximation} for convergence guarantee). In such scenarios, Nesterov's second method gives an accelerated version of projected gradient ascent/descent, such that each iterate lies in the domain $\overline{\mathcal{X}}$. This is a useful method to optimize convex functions of density matrices. When the Lipschitz constant is not known, a backtracking scheme can be used~\cite{tseng2010approximation}.
    \item We perform the ``outer" convex optimization over $\alpha$ in Eq.~\eqref{eqn:saddle_point_expression_optimization} using \texttt{scipy}'s \texttt{minimize\_scalar} routine. Through this optimization, we obtain the $\chi_1^*, \chi_2^* \in \mathcal{X}$ and $\alpha_* > 0$ components of the saddle-point.
    \item Using the so obtained $\chi_1^*, \chi_2^* \in \mathcal{X}$ and $\alpha_* > 0$, we find $\phi_* \in \mathcal{F}$ using Eq.~\eqref{eqn:phi_opt_quantum}.
\end{enumerate}

Once we have a saddle-point, the estimator can be easily computed using Eqs.~\eqref{eqn:JN_estimator} \& \eqref{eqn:JN_estimator_constant}. 

Note that we embed the Hermitian matrices into a real vector space before performing the above optimizations. This is possible because an isometric isomorphism exists between the set of Hermitian operators of a fixed size and a real vector space. We use $\epsZero  = 10^{-5}$ in the numerical implementation.

\subsection{Implementation with \texttt{cvxpy}\label{app:minimax_cvxpy}}
The optimal risk $\jnrisk$ as defined by the constrained optimization problem in Eq.~\eqref{eqn:JNriskprop3.1_general} can be evaluated directly using the \texttt{cvxpy} Python library, provided we take a logarithm of the constraint to make the problem convex, and multiply it by 2 to match the theory of Ref.~\cite{Juditsky2009}. This approach is faster than the above saddle point algorithm, at the cost of increased memory usage. A \texttt{cvxpy} implementation with a more balanced trade-off for larger dimensional problems is possible in the case of Pauli measurements by applying each POVM with a matrix-free algorithm. 

Regardless of such implementation details, however, one needs access to the optimal $\alpha_*$ from the saddle-point problem in Eq.~\eqref{eqn:Phi_quantum} to construct the estimator $\widehat{F}_*$. To do this without solving the saddle-point problem directly, we show that the parameter $\alpha$ of the saddle-point problem is actually the dual variable associated with our constraint. The optimal value of this dual variable is automatically returned by \texttt{cvxpy} alongside the optimal primal variables $\chi_1^*$ and $\chi_2^*$ as a consequence of the primal-dual algorithm with which it solves the problem.

Written in a standard form, Eq.~\eqref{eqn:JNriskprop3.1_general} becomes
\begin{align}
	\underset{\chi_1, \chi_2\,\in\, \mathcal{X}}{\text{minimize}} &\quad\quad \ip{\rho, \chi_2 - \chi_1}, \nonumber \\
	\text{subject to}&\quad\quad 2\ln(\epsilon/2) - 2\ln(\text{AffH}(A(\chi_1), A(\chi_2))) \leq 0, \label{eqn:JNriskprop3.1_stdform}
\end{align}
where the primal optimal value $p^*$ is related to the risk $\jnrisk$ by $-\tfrac{1}{2}p^* = \jnrisk$. Introducing $\widetilde{\alpha}$ as a dual variable, the Lagrangian for our nonlinear programming problem is given by
\begin{align*}
	\mathcal{L}(\chi_1, \chi_2; \widetilde{\alpha}) &= \langle \rho, \chi_2 - \chi_1\rangle\\
	 &\hspace{0.5cm}+\widetilde{\alpha}\Big(2\ln(\epsilon/2) - 2\ln(\text{AffH}(A(\chi_1), A(\chi_2)))\Big).
\end{align*}
This defines a concave Lagrange dual function $g: \mathbb{R} \to \mathbb{R}$ defined by
\begin{align*}
	g(\widetilde{\alpha}) = \min_{\chi_1, \chi_2\in \mathcal{X}} \mathcal{L}(\chi_1, \chi_2; \widetilde{\alpha}).
\end{align*}
Let $d^*$ denote the maximum value of the function $g$. Because our original problem in Eq.~\eqref{eqn:JNriskprop3.1_general} is convex and Slater's condition holds (as strict feasibility is achieved whenever $\chi_1 = \chi_2$), we have strong duality between the two problems. This means that $p^* = d^*$, and it follows that
\begin{align*}
	2\jnrisk &= -p^* = -d^* = -\max_{\widetilde{\alpha} \geq 0} g(\widetilde{\alpha})\\
	&= \min_{\widetilde{\alpha} \geq 0} \max_{\chi_1, \chi_2\in \mathcal{X}} -\mathcal{L}(\chi_1, \chi_2; \widetilde{\alpha})\\
	&= \min_{\widetilde{\alpha} \geq 0} \max_{\chi_1, \chi_1\in \mathcal{X}} \bigg\{\ip{\rho, \chi_1 - \chi_2}\\
	&\hspace{2.5cm} + 2\widetilde{\alpha}\bigg[\ln(\affh(A(\chi_1), A(\chi_2))) - \ln(\epsilon/2) \bigg]\bigg\}\\
	&= \inf_{\widetilde{\alpha} > 0} \bigg\{2 \widetilde{\alpha} \ln(2/\epsilon) + \max_{\chi_1, \chi_2 \in \mathcal{X}} \bigg[\ip{\rho, \chi_1} - \ip{\rho, \chi_2} \nonumber \\
	&\hspace{3.5cm}+ 2 \widetilde{\alpha} \ln(\affh(A(\chi_1), A(\chi_2)))\bigg]\bigg\}.
\end{align*}
With $\widetilde{\alpha} = \alpha$, this is exactly the expression in Eq.~\eqref{eqn:saddle_point_expression_optimization}, showing equivalence of the two values. Importantly, this means that the estimator found by solving the saddle-point problem in Eq.~\eqref{eqn:saddle_point_expression_optimization} can be equivalently constructed by finding both primal and dual optima of Eq.~\eqref{eqn:JNriskprop3.1_stdform}, which we compute using \texttt{cvxpy}.

\section{Minimax Method: Sample Complexity\label{app:minimax_sample_complexity}}
We begin by computing the best sample complexity that can be achieved by the minimax method. A detailed statement of this result is given in Theorem~\ref{thm:minimax_method_best_sample_complexity}. Below, we present a proof of this result.
\begin{proof}[Proof of Theorem~\ref{thm:minimax_method_best_sample_complexity}]
    \label{proof:minimax_method_best_sample_complexity}
    From Eq.~\eqref{eqn:JNriskprop3.1}, we know that the risk can be written as
    \begin{align*}
        &\jnrisk = \frac{1}{2} \max_{\chi_1, \chi_2 \in \mathcal{X}} \bigg\{\tr(\rho \chi_1) - \tr(\rho \chi_2)\ \bigg| \nonumber \\
                                                &\hspace{3cm} \prod_{l = 1}^L \left[F_C(\chi_1, \chi_2, \{E^{(l)}_k\})\right]^{R_l/2} \geq \frac{\failure}{2} \bigg\}
        \intertext{where}
        &F_C(\chi_1, \chi_2, \{E^{(l)}_k\}) = \left(\sum_{k = 1}^{N_l} \sqrt{\tr\left(E^{(l)}_k \chi_1\right) \tr\left(E^{(l)}_k \chi_2\right)}\right)^2
    \end{align*}
    is the classical fidelity between $\chi_1$ and $\chi_2$ determined by the POVM $\{E^{(l)}_k\}$. As noted in section~\ref{secn:minimax_method_optimal_risk}, we can write the fidelity between any two states as follows~\cite{fuchs1996distinguishability}.
    \begin{equation*}
        F(\chi_1, \chi_2) = \min_{\text{POVM } \{F_i\}} F_C(\chi_1, \chi_2, \{F_i\})
    \end{equation*}
    In particular, we have $F(\chi_1, \chi_2) \leq F_C(\chi_1, \chi_2, \{E^{(l)}_k\})$ for every POVM $\{E^{(l)}_k\}$ that we are using. Thus, we obtain the following lower bound on our risk
    \begin{align*}
        \jnrisk \geq \frac{1}{2} \max_{\chi_1, \chi_2 \in \mathcal{X}} \bigg\{&\tr(\rho \chi_1) - \tr(\rho \chi_2)\ \bigg| \nonumber \\
                                                        &F(\chi_1, \chi_2) \geq \left(\frac{\failure}{2}\right)^{\frac{2}{R}}\bigg\}
    \end{align*}
    where $R = \sum_{l = 1}^L R_l$ is the total number of shots. We now proceed to evaluating the lower bound. For convenience, we denote $\gamma = (\failure/2)^{2/R}$, so that the constraint becomes $F(\chi_1, \chi_2) \geq \gamma$. Next, we note that for a pure state $\rho$ and possibly mixed states $\chi_1$ and $\chi_2$, we have the following inequality for the fidelity in terms of the trace distance (see chapter 9 in Ref.~\cite{wilde2011classical})
    \begin{equation}
        \tr(\rho \chi_1) \leq \tr(\rho \chi_2) + \frac{1}{2} \norm{\chi_1 - \chi_2}_1
    \end{equation}
    where $\norm{\chi}_1$ is the Schatten $1$-norm of $\chi$. Then, using the Fuchs -- van de Graaf inequality $(1/2) \norm{\chi_1 - \chi_2}_1 \leq \sqrt{1 - F(\chi_1, \chi_2)}$ \cite{fuchs1999cryptographic}, we can write
    \begin{align}
        \tr(\rho \chi_1) - \tr(\rho \chi_2) &\leq \sqrt{1 - F(\chi_1, \chi_2)} \notag \\
                                            &\leq \sqrt{1 - \gamma} \label{eqn:fuchsvdg_ineq_constraint}
    \end{align}
    where the second line holds when $F(\chi_1, \chi_2) \geq \gamma$.

    We show that the upper bound in Eq.~\eqref{eqn:fuchsvdg_ineq_constraint} can be achieved by explicitly constructing the density matrices $\chi_1^*$ and $\chi_2^*$ achieving the maximum. For this purpose, we define $\Delta_\rho = \id - \rho$, and suppose that the dimension of the system is $d$ (i.e., $\rho$ is an $d \times d$ matrix). Then, let
    \begin{align*}
        \chi_1^* &= \frac{1 + \sqrt{1 - \gamma}}{2} \rho + \frac{1 - \sqrt{1 - \gamma}}{2} \frac{\Delta_\rho}{d - 1}, \\
        \chi_2^* &= \frac{1 - \sqrt{1 - \gamma}}{2} \rho + \frac{1 + \sqrt{1 - \gamma}}{2} \frac{\Delta_\rho}{d - 1}.
    \end{align*}
    Since $\rho$ is pure, there is some (normalized) vector $\ket{v_1}$ such that $\rho = \op{v_1}{v_1}$. Let $\{\ket{v_2}, \dotsc, \ket{v_d}\}$ be any orthonormal basis for the orthogonal complement of the subspace spanned by $\ket{v_1}$. Then, we can write $\Delta_\rho = \sum_{i = 2}^d \op{v_i}{v_i}$ using the resolution of identity. Therefore, in the basis $\{\ket{v_1}, \dotsc, \ket{v_d}\}$, the matrices $\chi_1^*$ and $\chi_2^*$ are diagonal. Since $\gamma < 1$, the diagonal entries of these matrices are real (and positive), and it is easy to check that they sum to $1$, showing that $\chi_1^*$ and $\chi_2^*$ are density matrices. Since they are diagonal, it is easy to compute the fidelity between them. We find that $F(\chi_1^*, \chi_2^*) = \gamma$, and therefore, these density matrices satisfy the constraint $F(\chi_1^*, \chi_2^*) \geq \gamma$. Further, we can see that these density matrices saturate the upper bound in Eq.~\eqref{eqn:fuchsvdg_ineq_constraint}, i.e., $\tr(\rho \chi_1^*) - \tr(\rho \chi_2^*) = \sqrt{1 - \gamma}$. Thus, we find that the lower bound on the risk is
    \begin{equation*}
        \jnrisk \geq \frac{1}{2} \sqrt{1 - \left(\frac{\failure}{2}\right)^{2/R}}
    \end{equation*}
    where we used $\gamma = (\failure/2)^{2/R}$.

    The inequality given in the above equation is tight: the POVM defined by $\{\rho, \Delta_\rho\}$ achieves the lower bound (see corollary \ref{corr:minimax_method_optimal_risk}). Thus, the best sample complexity given by the minimax method corresponding to a risk of $\jnrisk < 0.5$ and confidence level $1 - \failure \in (0.75, 1)$ is
\begin{align*}
    R &\geq \frac{2 \ln(2/\failure)}{\left|\ln(1 - 4\jnrisk^2)\right|} \\
      &\approx \frac{\ln(2/\failure)}{2\jnrisk^2} \text{ when } \jnrisk^2\ll 1.
\end{align*}
as noted.
\end{proof}

We now consider a family of two-outcome POVM measurements, and show that the risk given by the minimax method can be obtained by solving a one-dimensional optimization problem. Using this, we obtain a simple formula for an upper bound on the risk, and consequently, also a good bound on the sample complexity. In particular, this provides an upper bound on the sample complexity for the randomized Pauli measurement scheme given in Box~\hyperlink{box:pauli_measurement_scheme}{\ref*{secn:minimax_method}.1}.

\begin{theorem}
    \label{thm:minimax_sample_complexity_2outcomePOVM}
    Suppose we are given a pure target state $\rho$, and we perform $R$ repetitions of the POVM $\{\Theta, \Delta_\Theta\}$ defined as
    \begin{align*}
        \Theta &= \omega_1 \rho + \omega_2 \Delta_\rho \\
        \Delta_\Theta &= (1 - \omega_1) \rho + (1 - \omega_2) \Delta_\rho
    \end{align*}
    where $\Delta_\rho = \id - \rho$ and $\omega_1, \omega_2 \in [0, 1]$ are parameters satisfying $\omega_1 > \omega_2$. Also, define
    \begin{align*}
        \gamma &= \left(\frac{\failure}{2}\right)^{2/R}
        \intertext{and}
        R_o &= \frac{\ln(2/\failure)}{|\ln(\sqrt{\omega_1 \omega_2} + \sqrt{(1 - \omega_1)(1 - \omega_2)})|}
    \end{align*}

    \noindent Then, if $R > R_o$, the risk of the estimator given by the minimax method can be obtained by solving the one-dimensional optimization problem
    \begin{align}
        \jnrisk &= \frac{\sqrt{1 - \gamma}}{2(\omega_1 - \omega_2)} \max_{a \in \mathcal{A}_a} \sqrt{1 - (2a - 1)^2\gamma} \label{eqn:thm_sample_complexity_2outcomePOVM-risk}
        \intertext{where the set of allowed values for $a$ is given as}
        \mathcal{A}_a &= [0, 1] \cap \left((-\infty, a^{(1)}_-] \cup [a^{(1)}_+, \infty)\right) \nonumber \\
                                                                    &\hspace{2.5cm} \cap \left((-\infty, a^{(2)}_-] \cup [a^{(2)}_+, \infty)\right) \nonumber
        \intertext{with}
        a^{(1)}_{\pm} &= \omega_1 \pm \sqrt{\omega_1 (1 - \omega_1) \frac{(1 - \gamma)}{\gamma}} \nonumber \\
        a^{(2)}_{\pm} &= \omega_2 \pm \sqrt{\omega_2 (1 - \omega_2) \frac{(1 - \gamma)}{\gamma}}. \nonumber
    \end{align}
    For $R \leq R_o$, the risk is $\jnrisk = 0.5$.

    \noindent In particular, for any risk $\jnrisk \in (0, 0.5)$,
    \begin{align}
        R &\geq 2\frac{\ln\left(2/\failure\right)}{\left|\ln\left(1 - 4 (\omega_1 - \omega_2)^2 \jnrisk^2\right)\right|} \label{eqn:thm_sample_complexity_2outcomePOVM-repetitions} \\
          &\approx \frac{1}{2(\omega_1 - \omega_2)^2} \frac{\ln(2/\failure)}{\jnrisk^2} \nonumber
    \end{align}
    repetitions of the measurement are sufficient to achieve that risk with a confidence level of $1 - \failure \in (3/4, 1)$.
\end{theorem}
\begin{proof}
    For the case that we have $R$ repetitions of a single POVM $\{\Theta, \Delta_\Theta\}$, the risk can be written as (see Appendix~\ref{app:minimax_theory})
    \begin{align*}
        \jnrisk &= \frac{1}{2} \max_{\chi_1, \chi_2 \in \mathcal{X}} \bigg\{\tr(\rho \chi_1) - \tr(\rho \chi_2)\ \bigg| \nonumber \\
                                &\hspace{2cm}\text{AffH}(A(\chi_1), A(\chi_2)) \geq \sqrt{\gamma}\bigg\}
    \end{align*}
    where $A(\chi) = (\tr(\Theta \chi) + \epsZero/2) / (1 + \epsZero), \tr(\Delta_\Theta \chi) + \epsZero/2) / (1 + \epsZero))$. To begin with, we simplify this to a 2-dimensional optimization problem. For this purpose, we write the density matrices $\chi_1$, $\chi_2$ as a convex combination of the target state $\rho$ and some other trace-one Hermitian operator in the orthogonal complement of the subspace generated by $\rho$:
    \begin{align*}
        \chi_1 &= \alpha_1 \rho + (1 - \alpha_1) \rho_1^\perp \nonumber \\
        \chi_2 &= \alpha_2 \rho + (1 - \alpha_2) \rho_2^\perp
    \end{align*}
    with $\tr(\rho \rho_1^\perp) = \tr(\rho \rho_2^\perp) = 0$ and $0 \leq \alpha_1, \alpha_2 \leq 1$. Using this, the Hellinger affinity can be written as
    \begin{align}
        \text{AffH}(\alpha_1, \alpha_2) &= \frac{1}{1 + \epsZero} \left(\omega_1 \alpha_1 + \omega_2 (1 - \alpha_1) + \frac{\epsZero}{2}\right)^{1/2} \nonumber \\
                                        &\hspace{2cm} \left(\omega_1 \alpha_2 + \omega_2 (1 - \alpha_2) + \frac{\epsZero}{2}\right)^{1/2} \nonumber \\
                                        &\hspace{-1cm} + \frac{1}{1 + \epsZero} \left((1 - \omega_1) \alpha_1 + (1 - \omega_2) (1 - \alpha_1) + \frac{\epsZero}{2}\right)^{1/2} \nonumber \\
                                        &\hspace{0.8cm} \left((1 - \omega_1) \alpha_2 + (1 - \omega_2) (1 - \alpha_2) + \frac{\epsZero}{2}\right)^{1/2} \nonumber \\
                                        &\hspace{-1.8cm}\approx \sqrt{(\omega_2 + (\omega_1 - \omega_2) \alpha_1) (\omega_2 + (\omega_1 - \omega_2) \alpha_2)} \nonumber \\
                                        &\hspace{-1.5cm} + \sqrt{((1 - \omega_2) + (\omega_2 - \omega_1) \alpha_1) ((1 - \omega_2) + (\omega_2 - \omega_1) \alpha_2)} \label{eqn:affh_sample_complexity_2outcomePOVM}
    \end{align}
    Note that because of the parameter $\epsZero > 0$, the Hellinger affinity is differentiable. Since $\epsZero \ll 1$, we neglect it in Eq.~\eqref{eqn:affh_sample_complexity_2outcomePOVM} to prevent the equations from becoming cumbersome later. We can write the risk as
    \begin{align}
        2 \jnrisk &= \max_{\alpha_1, \alpha_2 \in [0, 1]} (\alpha_1 - \alpha_2) \nonumber \\
                                  &\hspace{1.1cm} \text{s.t. } -\ln(\affh(\alpha_1, \alpha_2)) \leq -\ln(\sqrt{\gamma}) \label{eqn:thm_sample_complexity_2outcomePOVM-risk_optimization}
    \end{align}
    We take a logarithm to make the optimization problem convex (see Proposition 3.1 in \cite{Juditsky2009}).
    
    Now, consider the case $R > R_o$, where $R_o$ is as defined in the statement of the theorem. Then, we argue that at the optimum, $\affh = \sqrt{\gamma}$. To see this, we first convert the above maximization to a minimization problem, and write its Lagrangian as
    \begin{align*}
        \mathcal{L} &= -\alpha_1 + \alpha_2 - \lambda \ln(\affh(\alpha_1, \alpha_2)) + \lambda \ln(\sqrt{\gamma}) \nonumber \\
                    &\hspace{1cm}- \nu^1_0 \alpha_1 + \nu^1_1 (\alpha_1 - 1) - \nu^2_0 \alpha_2 + \nu^2_1 (\alpha_2 - 1)
    \end{align*}
    where $\lambda, \nu^1_0, \nu^1_1, \nu^2_0, \nu^2_1$ are dual variables. At the optimum, the Karush-Kuhn-Tucker (KKT) conditions must be satisfied \cite{boyd2004convex}, which we list below for convenience.
    \begin{enumerate}
        \item (Primal feasibility) The ``primal" variables $\alpha_1, \alpha_2$ must lie in the domain $[0, 1]$.
        \item (Dual feasibility) The dual variables (corresponding to inequality constraints) $\lambda, \nu^1_0, \nu^1_1, \nu^2_0, \nu^2_1$ must be non-negative.
        \item (Complementary slackness) Either the dual variable must vanish or the constraint must be tight.
        \item (Stationarity) The gradient of the Lagrangian with respect to the primal variables must vanish.
    \end{enumerate}

    Since the gradient of $\mathcal{L}$ with respect to $\alpha_1$ and $\alpha_2$ must vanish at the optimum, we have
\begin{align*}
    \frac{\partial \mathcal{L}}{\partial \alpha_1} &= -1 - \lambda \frac{\partial \ln(\affh)}{\partial \alpha_1} - \nu^1_0 + \nu^1_1 = 0 \\
    \frac{\partial \mathcal{L}}{\partial \alpha_2} &= 1 - \lambda \frac{\partial \ln(\affh)}{\partial \alpha_2} - \nu^2_0 + \nu^2_1 = 0.
\end{align*}
Dual feasibility implies $\lambda, \nu^1_0, \nu^1_1, \nu^2_0, \nu^2_1 \geq 0$ at the optimum. If $\lambda = 0$ at the optimum, we must have $\nu^1_1 = 1 + \nu^1_0 > 0$ and $\nu^2_0 = 1 + \nu^2_1 > 0$. Then complementary slackness implies that $\affh \geq \sqrt{\gamma}$, $\alpha_1 = 1$ and $\alpha_2 = 0$ at the optimum. However, for $R > R_o$, we have
\begin{align*}
    \affh(\alpha_1 = 1, \alpha_2 = 0) &= \sqrt{\omega_1 \omega_2} + \sqrt{(1 - \omega_1)(1 - \omega_2)} \\
                                      &< \sqrt{\gamma}
\end{align*}
contradicting with the constraint. Thus, we must have $\lambda > 0$, implying that $\affh = \sqrt{\gamma}$ as claimed. Using this, we can reduce the problem to a one-dimensional problem that will eventually help perform the optimization. We do this by appropriately parametrizing each term in $\affh$:
\begin{align*}
    &\sqrt{(\omega_2 + (\omega_1 - \omega_2) \alpha_1) (\omega_2 + (\omega_1 - \omega_2) \alpha_2)} \equiv a \sqrt{\gamma} \\
    &\sqrt{((1 - \omega_2) + (\omega_2 - \omega_1) \alpha_1) ((1 - \omega_2) + (\omega_2 - \omega_1) \alpha_2)} \nonumber \\
    &\hspace{2cm} \equiv b \sqrt{\gamma} \\
    \intertext{Then, $\affh = \sqrt{\gamma}$ implies}
    &a + b = 1
\end{align*}
where $a, b \geq 0$. From the above equations, we can deduce that
\begin{align*}
    \alpha_1 + \alpha_2 &= \frac{(a^2 - b^2)}{\omega_1 - \omega_2} \gamma + \frac{(1 - 2\omega_2)}{\omega_1 - \omega_2} \\
    \alpha_1 \alpha_2 &= \frac{((1 - \omega_2) a^2 + \omega_2 b^2)}{(\omega_1 - \omega_2)^2} \gamma - \frac{\omega_2 (1 - \omega_2)}{(\omega_1 - \omega_2)^2}.
\end{align*}
These equations are well-defined because $\omega_1 > \omega_2$. Solving these simultaneous equations and applying the constraint $a + b = 1$, we obtain
\begin{align*}
    \alpha_1 &= \frac{(2a - 1)\gamma + (1 - 2\omega_2)}{2(\omega_1 - \omega_2)} + \frac{\sqrt{1 - \gamma}}{2(\omega_1 - \omega_2)} \sqrt{1 - (2a - 1)^2\gamma} \\
    \alpha_2 &= \frac{(2a - 1)\gamma + (1 - 2\omega_2)}{2(\omega_1 - \omega_2)} - \frac{\sqrt{1 - \gamma}}{2(\omega_1 - \omega_2)} \sqrt{1 - (2a - 1)^2\gamma}
\end{align*}
Since $a \in [0, 1]$, $(2a - 1)^2 \in [0, 1]$ and $\gamma \in (0, 1)$, the term in the square-root is non-negative, so $\alpha_1, \alpha_2$ are real. Furthermore, we have used the fact that $\alpha_1 \geq \alpha_2$ at the optimum because the risk involves maximization of $\alpha_1 - \alpha_2$; see Eq.~\eqref{eqn:thm_sample_complexity_2outcomePOVM-risk_optimization}.

Next, we need to impose the constraints $\alpha_1, \alpha_2 \in [0, 1]$. Requiring $\alpha_2 \geq 0$ (and thus $\alpha_1 \geq 0$) gives
\begin{align*}
    &(a - a^{(2)}_+) (a - a^{(2)}_-) \geq 0 \\
    &a^{(2)}_{\pm} = \omega_2 \pm \sqrt{\omega_2 (1 - \omega_2) \frac{(1 - \gamma)}{\gamma}}
\end{align*}
which means $a$ must lie in the region $(-\infty, a^{(2)}_-] \cup [a^{(2)}_+, \infty)$. Similarly, requiring $\alpha_1 \leq 1$ (and thus $\alpha_2 \leq 1$) gives
\begin{align*}
    &(a - a^{(1)}_+)(a - a^{(1)}_-) \geq 0 \\
    &a^{(1)}_{\pm} = \omega_1 \pm \sqrt{\omega_1 (1 - \omega_1) \frac{(1 - \gamma)}{\gamma}}
\end{align*}
which implies that $a$ must lie in the region $(-\infty, a^{(1)}_-] \cup [a^{(2)}_+, \infty)$. Therefore, the allowed values of $a$ are
\begin{align*}
    \mathcal{A}_a &= [0, 1] \cap \left((-\infty, a^{(2)}_-] \cup [a^{(2)}_+, \infty)\right) \nonumber \\
                  &\hspace{2cm} \cap \left((-\infty, a^{(1)}_-] \cup [a^{(1)}_+, \infty)\right)
\end{align*}
Note that the optimization problem defined by Eq.~\eqref{eqn:thm_sample_complexity_2outcomePOVM-risk_optimization} has a solution for all $R > 0$ (i.e., $\gamma \in (0, 1)$) because any $\alpha_1, \alpha_2 \in [0, 1]$ with $\alpha_1 = \alpha_2$ satisfies the constraints. Therefore, we must have $\mathcal{A}_a \neq \varnothing$.

Thus, the risk is given as
\begin{align*}
    \jnrisk &= \max_{a \in \mathcal{A}_a} \frac{1}{2} (\alpha_1 - \alpha_2) \\
                            &= \frac{\sqrt{1 - \gamma}}{2(\omega_1 - \omega_2)} \max_{a \in \mathcal{A}_a} \sqrt{1 - (2a - 1)^2\gamma}
\end{align*}
Now, for the case when $R \leq R_o$, we have $\sqrt{\gamma} \leq \sqrt{\omega_1 \omega_2} + \sqrt{(1 - \omega_1)(1 - \omega_2)} = \affh(\alpha_1 = 1, \alpha_2 = 0)$. Therefore, $\alpha_1 = 1$ and $\alpha_2 = 0$ satisfy the constraint of the optimization in Eq.~\eqref{eqn:thm_sample_complexity_2outcomePOVM-risk_optimization}, giving $\jnrisk = 0.5$.

The last part of the statement of the theorem follows from the observation that Eq.~\eqref{eqn:thm_sample_complexity_2outcomePOVM-repetitions} implies $R > R_o$ when $\jnrisk < 0.5$, and for $R > R_o$, we have the inequality $\jnrisk \leq \sqrt{1 - \gamma}/(2(\omega_1 - \omega_2))$.
\end{proof}

Note that there is no loss of generality in requiring that $\omega_1 > \omega_2$, for if $\omega_2 < \omega_1$, we can simply swap $\Theta$ and $\Delta_\Theta$. When $\omega_1 = \omega_2$, we have $\Theta = \Delta_\Theta = \id / 2$, which means we learn nothing about the state. Indeed $R_o \to \infty$ as $\omega_1 \to \omega_2$, alluding to this fact. On the other hand, the best we can do is when $\omega_1$ and $\omega_2$ are farthest from each other, and this leads to the following result.

\begin{corollary}
    \label{corr:minimax_method_optimal_risk}
    Let $\rho$ be any pure target state, and $\Delta_\rho = \id - \rho$. Then, for $R$ repetitions of the POVM $\{\rho, \Delta_\rho\}$, the estimator given by minimax method achieves the risk
    \begin{equation*}
        \jnrisk = \frac{1}{2} \sqrt{1 - \left(\frac{\failure}{2}\right)^{2/R}}
    \end{equation*}
\end{corollary}
\begin{proof}
    We have $\omega_1 = 1$ and $\omega_2 = 0$. Substituting this in Theorem \ref{thm:minimax_sample_complexity_2outcomePOVM}, we can see that $a^{(1)}_{\pm} = 1$ and $a^{(2)}_{\pm} = 0$. Therefore, the allowed values of $a$ are $\mathcal{A}_a = [0, 1]$. Subsequently, the risk is given as
    \begin{align*}
        \jnrisk &= \frac{\sqrt{1 - \gamma}}{2} \max_{a \in [0, 1]} \sqrt{1 - (2a - 1)^2\gamma} \\
                                &= \frac{\sqrt{1 - \gamma}}{2}
    \end{align*}
    as claimed.
\end{proof}

We also consider a more restricted family of POVMs that are relevant to the stabilizer measurements described in section \ref{secn:minimax_method_stabilizer_states}.
\begin{corollary}
    \label{corr:minimax_sample_complexity_stabilizer}
    Suppose we are given a pure target state $\rho$, and we perform $R$ repetitions of the POVM $\{\Theta, \Delta_\Theta\}$ defined as
    \begin{align*}
        \Theta &= \rho + \frac{\xi/2 - 1}{\xi - 1} \Delta_\rho \\
        \Delta_\Theta &= \frac{\xi/2}{\xi - 1} \Delta_\rho
    \end{align*}
    where $\Delta_\rho = \id - \rho$ and $\xi \geq 2$ is a parameter. Also, define
    \begin{equation*}
        R_o = 2 \frac{\ln(2/\failure)}{\ln\left(\frac{\xi - 1}{\xi/2 - 1}\right)}.
    \end{equation*}

    \noindent Then, if $R > R_o$, the risk of the estimator given by the minimax method is
    \begin{align}
        \jnrisk &= \begin{cases} \left(\frac{\xi - 1}{\xi}\right) \sqrt{1 - \gamma} & b_- \geq 1 \\
                                                 \left(\frac{\xi - 1}{\xi}\right) (1 - \gamma) \sqrt{1 + b_-(2 - b_-) \left(\frac{\gamma}{1 - \gamma}\right)} & b_- < 1,\\ &\hspace{-1.5cm}|b_- - 1| \leq |b_+ - 1| \\
                                                 \left(\frac{\xi - 1}{\xi}\right) (1 - \gamma) \sqrt{1 + b_+(2 - b_+) \left(\frac{\gamma}{1 - \gamma}\right)} & b_- < 1,\\ &\hspace{-1.5cm}|b_- - 1| > |b_+ - 1|
                                   \end{cases} \label{eqn:thm_sample_complexity-risk}
        \intertext{where}
        \gamma &= \left(\frac{\failure}{2}\right)^{2/R} \notag
        \intertext{and}
        b_{\pm} &= \left(\frac{\xi}{\xi - 1}\right) \left[1 \pm \sqrt{\left(\frac{1 - \gamma}{\gamma}\right) \left(\frac{\xi - 2}{\xi}\right)}\right] \nonumber
    \end{align}
    For $R \leq R_o$, the risk is $\jnrisk = 0.5$.

    \noindent In particular, for any risk $\jnrisk \in (0, 0.5)$,
    \begin{align}
        R &\geq 2\frac{\ln\left(2/\failure\right)}{\left|\ln\left(1 - \left(\frac{\xi}{\xi - 1}\right)^2 \jnrisk^2\right)\right|} \label{eqn:thm_sample_complexity-repetitions} \\
          &\approx 2 \left(\frac{\xi - 1}{\xi}\right)^2 \frac{\ln(2/\failure)}{\jnrisk^2} \nonumber
    \end{align}
    repetitions of the measurement are sufficient to achieve that risk with a confidence level of $1 - \failure \in (3/4, 1)$.
\end{corollary}
\begin{proof}
    In the context of Theorem \ref{thm:minimax_sample_complexity_2outcomePOVM}, we have
    \begin{equation*}
        \omega_1 = 1 \quad \text{and} \quad \omega_2 = \frac{\xi/2 - 1}{\xi - 1}
    \end{equation*}
    This implies $a^{(1)}_{\pm} = 1$, and also
    \begin{equation*}
        a^{(2)}_{\pm} = \frac{\xi/2 - 1}{\xi - 1} \pm \frac{\xi/2}{\xi - 1} \sqrt{\frac{(1 - \gamma)}{\gamma} \frac{\xi - 2}{\xi}}
    \end{equation*}
    For convenience, we perform the change of variables $b = 2(1 - a)$, and define
    \begin{align*}
        b_{\pm} &\equiv 2(1 - a^{(2)}_{\mp}) \\
                &= \left(\frac{\xi}{\xi - 1}\right) \left[1 \pm \sqrt{\left(\frac{1 - \gamma}{\gamma}\right) \left(\frac{\xi - 2}{\xi}\right)}\right]
    \end{align*}
    Clearly, $b_- \leq b_+$. Note that $R > R_o$ implies $b_- > 0$, which means $a_+ < 1$. Thus, we have that $\mathcal{A}_a = [0, 1] \cap ((-\infty, a^{(2)}_-] \cup [a^{(2)}_+, \infty)) = [0, a^{(2)}_-] \cup [a^{(2)}_+, 1]$ is non-empty. With respect to the variable $b$, these allowed values can be expressed as $\mathcal{A}_b = [0, b_-] \cup [b_+, 2]$. The risk is then given as
    \begin{align*}
        \jnrisk &= \left(\frac{\xi - 1}{\xi}\right) \sqrt{1 - \gamma} \max_{b \in \mathcal{A}_b} \sqrt{1 - (1 - b)^2 \gamma} \notag \\
                                &= \left(\frac{\xi - 1}{\xi}\right) (1 - \gamma) \max_{b \in \mathcal{A}_b} \sqrt{1 + b (2 - b) \left(\frac{\gamma}{1 - \gamma}\right)}. \notag
    \end{align*}
    Since the objective of maximization is symmetric about $b = 1$ and the maximum is achieved at $b = 1$, the allowed value of $b$ closest to $1$ achieves the maximum. Noting that $b_+ > 1$, we obtain the expression given in the statement of the corollary.
\end{proof}

Finally, we obtain a bound on the sample complexity of the randomized Pauli measurement scheme described in Box \hyperlink{box:pauli_measurement_scheme}{\ref*{secn:minimax_method}.1}. The statement is given in Theorem~\ref{thm:minimax_method_pauli_scheme_sample_complexity}, so we simply give a proof.

\begin{proof}[Proof of Theorem~\ref{thm:minimax_method_pauli_scheme_sample_complexity}]
\label{proof:minimax_method_pauli_scheme_sample_complexity}
    For convenience, we reproduce below the measurement strategy given in Box \hyperlink{box:pauli_measurement_scheme}{\ref*{secn:minimax_method}.1}. Given a target state $\rho$,
    \begin{enumerate} 
        \item Sample a (non-identity) Pauli operator $W_i$ with probability
        \begin{equation*}
            p_i = \frac{|\tr(W_i \rho)|}{\sum_{i = 1}^{d^2 - 1} |\tr(W_i \rho)|}, \quad i = 1, \dotsc, d^2 - 1,
        \end{equation*}
        and record the outcome ($\pm 1$) of the measurement.
        \item Flip the measurement outcome (i.e, $\pm 1 \to \mp 1$) if $\tr(\rho W_i) < 0$, else retain the original measurement outcome
        \item Repeat this procedure $R$ times.
    \end{enumerate}
    Note that measuring the Pauli $W_i$ and flipping the measurement outcome is equivalent to measuring the operator $S_i~=~\text{sign}(\tr(W_i \rho)) W_i$.
    
    We can describe the above measurement strategy using the effective POVM described below, which we obtain by finding a positive semidefinite operator that reproduces the measurement statistics. The probability of obtaining $+1$ outcome can be written as
    \begin{align*}
        \text{Pr}(+1) &= \sum_{i = 1}^{d^2 - 1} (\text{Pr. choosing $S_i$}) \nonumber \\
                      &\hspace{1cm} (\text{Pr. outcome $1$ upon measuring $S_i$}) \\
                      &\equiv \tr(\Theta \sigma) \\
        \Theta &= \sum_{i = 1}^{d^2 - 1} p_i \mathbb{P}^+_i \\
        \mathbb{P}^+_i &= \frac{\id + S_i}{2}
    \end{align*}
    where $S_i = \text{sign}(\tr(W_i \rho)) W_i$. Substituting for $p_i$ and a simple rearrangement of terms gives
    \begin{align*}
        \Theta &= \frac{\id}{2} + \frac{d}{2\mathcal{N}} \sum_{i = 1}^{d^2 - 1} \frac{|\tr(W_i \rho)|}{d} S_i \\
               &= \frac{(d + (\mathcal{N} - 1))}{2\mathcal{N}} \rho + \frac{(\mathcal{N} - 1)}{2\mathcal{N}} \Delta_\rho.
    \end{align*}
    $\Delta_\Theta = \id - \Theta$ is given as
    \begin{equation*}
        \Delta_\Theta = \frac{((\mathcal{N} + 1) - d)}{2\mathcal{N}} \rho + \frac{(\mathcal{N} + 1)}{2\mathcal{N}} \Delta_\rho.
    \end{equation*}
    The effective POVM is then $\{\Theta, \Delta_\Theta\}$. Substituting
    \begin{equation*}
        \omega_1 = \frac{(d + (\mathcal{N} - 1))}{2\mathcal{N}} \quad \text{and} \quad \omega_2 = \frac{(\mathcal{N} - 1)}{2\mathcal{N}}
    \end{equation*}
    in Theorem \ref{thm:minimax_sample_complexity_2outcomePOVM}, we obtain Eq.~\eqref{eqn:minimax_method_pauli_scheme_sample_complexity}.

    Now, note that for any pure state $\rho$, we have $\tr(\rho^2) = 1$ and this gives us
    \begin{equation}
        \sum_{i = 1}^{d^2 - 1} (\tr(W_i \rho))^2 = d - 1. \label{eqn:pure_state_pauli_weights_constraint}
    \end{equation}
    Therefore, to obtain an upper bound on $\mathcal{N}$, we solve the following optimization problem. Let $M$ be a positive integer and $\beta > 0$ be any real number such that $M > \beta$. We solve
    \begin{align*}
        \max\ &\sum_{i = 1}^M x_i \nonumber \\
        \text{s.t. } &x_i \in [0, 1]\quad \forall i = 1, \dotsc, M \nonumber \\
                     &\sum_{i = 1}^M x_i^2 \leq \beta
    \end{align*}
    The above problem gives a bound for the special case of $x_i = |\tr(W_i \rho)|$, $M = d^2 - 1$ and $\beta = d - 1$, while saving the trouble of optimizing over all density matrices. Note that we consider the relaxation $\sum_{i = 1}^M x_i^2 \leq \beta$, instead of the original equality constraint $\sum_{i = 1}^M x_i^2 = \beta$ because quadratic equality constraints do not define a convex set in general. Such a relaxation is inconsequential because equality is attained at the optimum. It can be shown using the KKT conditions that the optimum corresponds to
    \begin{align*}
        x_i &= \sqrt{\frac{\beta}{M}}\quad \forall i \in \{1, \dotsc, M\}\notag \\
        \implies \sum_{i = 1}^M x_i &= \sqrt{M \beta}. \notag
    \end{align*}

    Returning to the problem of Pauli weights, since $M = d^2 - 1$ and $\beta = d - 1$, we obtain $\mathcal{N} \leq \sqrt{(d^2 - 1)(d - 1)} = \sqrt{d + 1} (d - 1)$ as claimed. Substituting in the expression for sample complexity gives the desired upper bound.
\end{proof}

In contrast to the above result which shows that a good sample complexity can be obtained for the randomized Pauli measurement scheme, we consider the case of a bad measurement protocol. Namely, we are given an $n$-qubit stabilizer state and we measure $n - 1$ of its generators, where the measurements are subspace measurements. Then, as noted in Proposition~\ref{prop:minimax_method_stabilizer_insufficient_measurements}, the minimax method gives a risk of $0.5$. Here, we present a proof for this statement.
\begin{proof}[Proof of Proposotion~\ref{prop:minimax_method_stabilizer_insufficient_measurements}]
    \label{proof:minimax_method_stabilizer_insufficient_measurements}
    Let $\rho$ be an $n$-qubit stabilizer state generated by $S_1, \dotsc, S_n$. The measurement protocol corresponds to measuring only the first $n - 1$ generators $S_1, \dotsc, S_{n - 1}$. The measurement of $S_l$ has the POVM $\{E^{(l)}_1, E^{(l)}_2\}$ where $E^{(l)}_1$ is the projection on the $+1$ eigenspace of $S_l$ while $E^{(l)}_2$ is the projection on $-1$ eigenspace of $S_l$, for $l = 1, \dotsc, n - 1$. Suppose that the $l^{\text{th}}$ measurement is repeated $R_l$ times.

    From Eq.~\eqref{eqn:JNriskprop3.1}, we know that the risk of the minimax method can be written as
    \begin{align*}
        &\jnrisk = \frac{1}{2} \max_{\chi_1, \chi_2 \in \mathcal{X}} \bigg\{\tr(\rho \chi_1) - \tr(\rho \chi_2)\ \bigg| \nonumber \\
                                                &\hspace{3cm} \prod_{l = 1}^{n - 1} \left[F_C(\chi_1, \chi_2, \{E^{(l)}_k\})\right]^{R_l/2} \geq \frac{\failure}{2} \bigg\}
        \intertext{where}
        &F_C(\chi_1, \chi_2, \{E^{(l)}_k\}) = \left(\sum_{k = 1}^2 \sqrt{\tr\left(E^{(l)}_k \chi_1\right) \tr\left(E^{(l)}_k \chi_2\right)}\right)^2
    \end{align*}
    is the classical fidelity corresponding to the POVM $\{E^{(l)}_1, E^{(l)}_2\}$ for $l = 1, \dotsc, n - 1$. Our strategy is to construct two density matrices $\chi_1$ and $\chi_2$ that satisfy the constraints of the optimization defining the risk, such that the value of the risk is $0.5$.

    To that end, let $\tilde{\rho}$ be the stabilizer state generated by $S_1, \dotsc, -S_n$, where, the last generator of $\tilde{\rho}$ differs from that of $\rho$ by a negative sign. Note that the states $\rho$ and $\tilde{\rho}$ are orthogonal to each other. Observe that the classical fidelity between the states $\rho$ and $\tilde{\rho}$ corresponding to the POVM $\{E^{(l)}_1, E^{(l)}_2\}$ is $F_C(\rho, \tilde{\rho}, \{E^{(l)}_1, E^{(l)}_2\}) = 1$ for all measured stabilizers since $\tr(E^{(l)}_1 \rho) = \tr(E^{(l)}_1 \tilde{\rho}) = 1$ while $\tr(E^{(l)}_2 \rho) = \tr(E^{(l)}_2 \tilde{\rho}) = 0$ for all $l = 1, \dotsc, n - 1$.

    Thus, taking $\chi_1 = \rho$ and $\chi_2 = \tilde{\rho}$, we find that the risk is $\jnrisk = 0.5$, which is the maximum possible value for the risk. In other words, when an insufficient number of stabilizer measurements are provided, the minimax method infers that the fidelity cannot be estimated accurately.
\end{proof}

\section{Comparison with QSV\label{app:qsv_comparison}}
We begin by giving a brief overview of the measurement strategy used in QSV~\cite{pallister2018optimal}.
QSV assumes that either $F(\rho, \sigma) = 1$  or $F(\rho, \sigma) \leq 1 - \qsverr$ holds for a fixed $\qsverr \in (0, 1)$.
The goal is to reject the hypothesis $F(\rho, \sigma) \leq 1 - \qsverr$.
To that end, QSV protocol assumes that we have access to a set of operators $\mathcal{S} = \{P_1, \dotsc, P_L\}$ with $0 \leq P_i \leq \id$ describing two-outcome POVM $\{P_i, \id - P_i\}$.
Pallister \textit{et al.}~\cite{pallister2018optimal} assume in addition that $\tr(P_i \rho) = 1$ for all $i$.
The protocol proceeds by randomly sampling the operator $P_i$ with probability $\mu_i$ and measuring the POVM $\{P_i, \id - P_i\}$.
If the outcome corresponding to $P_i$ is observed, then one repeats this procedure.
On the other hand, if one observes the outcome corresponding to $\id - P_i$, then one stops the procedure and declares ``fail".
If after many repetitions of this procedure the protocols does not fail, one can conclude with high probability that $\tr(\rho \sigma) = 1$.

The above measurement protocol can be described by the effective POVM $\{\Omega, \id - \Omega\}$ with $\Omega = \sum_{i = 1}^L \mu_i P_i$ which satisfies $\tr(\Omega \rho) = 1$.
QSV searches for an optimal measurement strategy $\Omega$ by minimizing the probability of wrongly declaring ``pass" in the worst-case scenario:
\begin{equation}
    \min_{\Omega} \max_{\substack{\sigma\\ F(\rho, \sigma) \leq 1 - \qsverr}} \tr(\Omega \sigma)
    \defeq 1 - \Delta_{\qsverr}.
\end{equation}
Consequently, when running the protocol with $R$ states, the false acceptance probability for such an optimal $\Omega$ is bounded above by $(1 - \Delta_{\qsverr})^R$.
Since the probability of wrongly declaring ``fail" is zero by the assumption that $\tr(\rho \Omega) = 1$, we can conclude
\begin{equation}
    R \geq \frac{\ln(1/\failure)}{\left|\ln\left(1 - \Delta_{\qsverr}\right)\right|} \approx \frac{\ln(1/\failure)}{\Delta_{\qsverr}}. \label{eqn:qsv_sample_complexity}
\end{equation}
measurement outcomes are sufficient to certify that $F(\rho, \sigma) > 1 - \qsverr$ with probability $1 - \failure$.
Note that such a measurement protocol is minimax optimal in the sense that one finds the best measurement strategy $\Omega$ for the worst possible state $\sigma$
satisfying $F(\rho, \sigma) \leq 1 - \qsverr$, under the POVM options $\mathcal{S}$ and the QSV assumption that $F(\rho,\sigma) \not\in(1-\qsverr,1)$.

The above description shows that the minimax approach of QSV bears similarities with our approach towards fidelity estimation.
Below, we discuss two differences between QSV and our method (in addition to those already pointed out in section~\ref{secn:qsv_comparison}).

First, we note that the assumption that the state $\sigma$ prepared in the lab either satisfies $F(\rho, \sigma) = 1$ or $F(\rho, \sigma) \leq 1 - \qsverr$ might be too stringent in practice. For example, if we reject $F(\rho, \sigma) \leq 1 - \qsverr$, we can't automatically conclude that the fidelity is $1$.
Our study makes no such assumption, which makes it more practically amenable.
As mentioned earlier, some studies on QSV have relaxed this assumption, but at the cost of increasing the sample complexity compared to the original QSV guarantees~\cite{jiang2020towards}. 

Second, the assumption that the measurement operators $P_i$ satisfy $\tr(\Omega P_i) = 1$ can be restrictive in practice, even though they allow for designing optimal measurement protocols in theory. For example, this assumption does not hold if one wishes to estimate (or certify) the fidelity of a $W$-state using Pauli measurements. This is a practically relevant problem which can be tackled using our approach~\cite{PRL}. Moreover, our fidelity estimation method is not restricted to protocols involving random sampling of measurement settings, which is again important for practical applications. That said, when one presents an optimal protocol involving random sampling, we can usually adjust the sampling probabilities so that our fidelity estimation method provides similar optimality results. This was shown, for example, in Sec.~\ref{secn:minimax_method_RPM_scheme}, which matches the performance guarantees of DFE.

We show that a similar result can be obtained with our method for estimating the fidelity of two-qubit states using the measurement protocol given by Pallister \textit{et al.}~\cite{pallister2018optimal} without changing the sampling probabilities.
Using the equivalence of states under rotation by local unitaries, any two-qubit state can be expressed in the form $\ket{\psi} = \sin(\theta) \ket{00} + \cos(\theta) \ket{11}$ for $\theta \in [0, \pi/2]$~\cite{pallister2018optimal}.
Then, the optimal measurement strategy for QSV can be expressed as sampling the projectors $\{P_1, \dotsc, P_4\}$ with probability $\mu_1 = \alpha(\theta)$, $\mu_2 = \mu_3 = \mu_4 = (1 - \alpha(\theta))/3$ ($\theta \in (0, \pi/2) \setminus \{\pi/4\}$) and measuring them~\cite{pallister2018optimal}. The projectors are given by $P_1 = \op{00}{00} + \op{11}{11}$ and $P_i = \id - \op{\phi_i}{\phi_i}$ for $i = 2, 3, 4$, where the states $\ket{\phi_i}$ and the number $\alpha(\theta)$ are given in Eq.~\eqref{eqn:qsv_twoqubit_projectors_prob}. Pallister \textit{et al.}~\cite{pallister2018optimal} show that
\begin{equation}
    R \approx (2 + \sin(\theta) \cos(\theta)) \frac{\ln(1/\failure)}{\qsverr} \label{eqn:qsv_twoqubit_sample_complexity}
\end{equation}
measurements of the operator $\Omega = \sum_{i = 1}^4 \mu_i P_i$ suffice to reject $F(\rho, \sigma) \leq 1 - \qsverr$ with a confidence level of $1 - \failure$. 
The following result shows that using this measurement protocol, our method performs fidelity estimation in an optimal manner.
\begin{proposition}
    Let $\ket{\psi} = \sin(\theta) \ket{00} + \cos(\theta) \ket{11}$, where $\theta \in (0, \pi/2)$, $\theta \neq \pi/4$, denote any two-qubit state up to rotation by local unitaries. Consider the measurement protocol where the projectors $P_1, P_2, P_3, P_4$ are sampled as per probability $\mu_1 = \alpha(\theta)$, $\mu_2 = \mu_3 = \mu_4 = (1 - \alpha(\theta))/3$ and their outcomes are recorded. Here, $P_1 = \op{00}{00} + \op{11}{11}$, $P_i = \id - \op{\phi_i}{\phi_i}$ for $i = 2, 3, 4$, with
    \begin{align}
        \ket{\phi_1} &= \left(\frac{1}{\sqrt{1+\tan\theta}}\ket{0} + \frac{e^{\frac{2\pi i}{3}}}{\sqrt{1+\cot\theta}}\ket{1} \right) \nonumber \\
                     &\quad\otimes \left(\frac{1}{\sqrt{1+\tan\theta}}\ket{0} + \frac{e^{\frac{\pi i}{3}}}{\sqrt{1+\cot\theta}}\ket{1} \right), \nonumber \\
        \ket{\phi_2} &= \left(\frac{1}{\sqrt{1+\tan\theta}}\ket{0} + \frac{e^{\frac{4\pi i}{3}}}{\sqrt{1+\cot\theta}}\ket{1} \right) \nonumber \\
                     &\quad\otimes \left(\frac{1}{\sqrt{1+\tan\theta}}\ket{0} + \frac{e^{\frac{5\pi i}{3}}}{\sqrt{1+\cot\theta}}\ket{1} \right), \nonumber \\
        \ket{\phi_3} &= \left(\frac{1}{\sqrt{1+\tan\theta}}\ket{0} + \frac{1}{\sqrt{1+\cot\theta}}\ket{1} \right) \nonumber \\
                     &\quad\otimes \left(\frac{1}{\sqrt{1+\tan\theta}}\ket{0} - \frac{1}{\sqrt{1+\cot\theta}}\ket{1} \right), \textnormal{ and } \nonumber \\
        \alpha(\theta) &= \frac{2 - \sin(2 \theta)}{4 + \sin(2 \theta)}. \label{eqn:qsv_twoqubit_projectors_prob}
    \end{align}
    Then,
    \begin{align}
        R &\geq 2 \frac{\ln(2/\failure)}{\left|\ln\left(1 - \frac{4}{(2 + \sin(\theta) \cos(\theta))^2}  \jnrisk^2\right)\right|} \nonumber \\
          &\approx \frac{(2 + \sin(\theta) \cos(\theta))^2}{2} \frac{\ln(2/\failure)}{\jnrisk^2} \textnormal{ when } \jnrisk \ll 1 \label{eqn:sample_complexity_twoqubit_states}
    \end{align}
    repetitions of the measurement protocol is sufficient to reach attain a risk of $\jnrisk$ with a confidence level of $1 - \failure \in (0.75, 1)$.
\end{proposition}
\begin{proof}
    The effective measurement operator is given by~\cite{pallister2018optimal}
    \begin{align*}
        \Omega &= \alpha P_1 + \frac{(1 - \alpha)}{3} \sum_{i = 2}^4 P_i \\
               &= \alpha P_1 + (1 - \alpha) \Omega_3,
    \end{align*}
    where
    \begin{equation*}
        \Omega_3 = \id - \frac{1}{(1 + t)^2} \begin{pmatrix}
                        1 & 0 & 0 & -t \\
                        0 & t & 0 & 0 \\
                        0 & 0 & t & 0 \\
                        -t & 0 & 0 & t^2
                      \end{pmatrix}
    \end{equation*}
    and $t = \tan(\theta)$. By diagonalizing $\Omega$, we find that
    \begin{equation*}
        \Omega = \rho + \frac{2 + \sin(2 \theta)}{4 + \sin(2 \theta)} (\id - \rho).
    \end{equation*}
    Then, from Thm.~\ref{thm:minimax_sample_complexity_2outcomePOVM}, Eq.~\eqref{eqn:sample_complexity_twoqubit_states} follows. We remark that our risk is invariant under unitary rotations of the target state $\rho \to U \rho U^\dagger$ and the POVM $E_k \to U E_k U^\dagger$ as seen from Eq.~\eqref{eqn:JNriskprop3.1}. Therefore, the above result essentially holds for any two-qubit state.
\end{proof}
We exclude the points $\theta = 0, \pi/4, \pi/2$ in the above result as either $\tan(\theta)$ or $\cot(\theta)$ become undefined at these points.
For $\theta = 0, \pi/2$, we get a product state, and therefore, one can directly measure Pauli $ZZ$.
For $\theta = \pi/4$, we get a stabilizer state, for which we have an optimal measurement strategy (see section~\ref{secn:minimax_method_stabilizer_states}).
These observations were also made by Pallister \textit{et al}~\cite{pallister2018optimal}.

Note that the number of measurements given in Eq.~\eqref{eqn:sample_complexity_twoqubit_states} required by our method is similar to Eq.~\eqref{eqn:qsv_twoqubit_sample_complexity}, except for a pre-factor of order $1$ and scaling of $1/\jnrisk^2$ with the risk $\jnrisk$. This scaling is essentially optimal, as seen from Thm.~\ref{thm:minimax_method_best_sample_complexity}, and thus, we can use their measurement protocol along with our fidelity estimation method to optimally estimate the fidelity of two-qubit states. This reinforces the expectation that good measurement schemes lead to good sample complexity using our method. Conversely, if our method cannot give a good sample complexity for some measurement protocol, then owing to the definition of minimax optimal risk and Eq.~\eqref{eqn:risk_minimax_guarantee}, we can infer that no estimation method can give a good sample complexity for that measurement protocol under the same assumptions as our method.
\end{document}